\documentclass[journal]{IEEEtran}

\usepackage{multicol}

\usepackage{graphicx}
\usepackage{amsmath,amsfonts,amssymb}
\usepackage{algorithmic}
\usepackage[linesnumbered,ruled]{algorithm2e}

\usepackage{ntheorem}
\theoremheaderfont{\bfseries\upshape}
\theoremseparator{:}
\theorembodyfont{\upshape}
\newtheorem{theorem}{Theorem}
\newtheorem{lemma}{Lemma}
\newtheorem{definition}{Definition}
\newtheorem{proof}{Proof}
\newtheorem{proposition}{Proposition}
\newtheorem{remark}{Remark}

\begin{document}
%
\title{Polar-Slotted ALOHA over Slot Erasure Channels}

\author{
        Zhijun~Zhang,~\IEEEmembership{Student Member,~IEEE,}
        Kai~Niu,~\IEEEmembership{Member,~IEEE,}
        Jincheng Dai,~\IEEEmembership{Member,~IEEE,}
        and Chao~Dong,~\IEEEmembership{Member,~IEEE}~
\thanks{This work is supported by the National Key R$\&$D Program of China (2018YFE0205501)
and the National Natural Science Foundation of China (Grant No.61671080). (Corresponding author: Kai Niu.)}

\thanks{The authors are with the Key Laboratory of Universal Wireless Communications,
Ministry of Education, Beijing University of Posts and Telecommunications,
Beijing 100876, China. e-mail: zjzhang@bupt.edu.cn; niukai@bupt.edu.cn; daijincheng@bupt.edu.cn and dongchao@bupt.edu.cn
}
}


\maketitle

\begin{abstract}
In this paper, we design a new polar slotted ALOHA (PSA) protocol over the slot erasure channels, which uses polar coding to construct the identical slot pattern (SP) assembles within each active user and base station. A theoretical analysis framework for the PSA is provided. First, by using the packet-oriented operation for the overlap packets when they conflict in a slot interval, we introduce the packet-based polarization transform and prove that this transform is independent of the packet's length. Second, guided by the packet-based polarization, an SP assignment (SPA) method with the variable slot erasure probability (SEP) and a SPA method with a fixed SEP value are designed for the PSA scheme. Then, a packet-oriented successive cancellation (pSC) and a pSC list (pSCL) decoding algorithm are developed. Simultaneously, the finite-slots throughput bounds and the asymptotic throughput for the pSC algorithm are analyzed. The simulation results show that the proposed PSA scheme can achieve an improved throughput with the pSC/SCL decoding algorithm over the traditional repetition slotted ALOHA scheme.
\end{abstract}

\begin{IEEEkeywords}
 Slotted ALOHA, polar code, slot erasure channel, successive cancellation list decoding.
\end{IEEEkeywords}

\IEEEpeerreviewmaketitle

\section{Introduction}

Motivated by the critical need to better support massive machine to machine communication in the upcoming cellular communications, contention resolution diversity slotted ALOHA (CRDSA) \cite{CRDSA}, irregular repetition slotted ALOHA (IRSA) \cite{IRSA} and coded slotted ALOHA (CSA) \cite{CSA:Paolini} were proposed to enhance the throughput of uncoordinated random access schemes by using the iterative successive interference cancellation (SIC) technique to resolve packet collisions. In slotted ALOHA schemes, a binary vector, called as the user¡¯s slot pattern (SP), is used to denote the slot positions whereby the copies of the user information packet will be transmitted within these slots and marked them as '$1$'s, otherwise marked as '$0$'s. The number of non-zero elements in the vector is denoted as SP¡¯s weight. For example, a user's SP is $(1,0,0,0,0,0,1,0)$ which means that the $1$st and $7$th slots are used to transmit the user's packet copies in a slot-frame. It is well known that the CRDSA scheme uses an identical $1/2$-rate repetition encoding of the information packet for each active user, that is, guided by optimized weight-$2$ SPs, each user transmits twice copies within a slot-frame simultaneously. Compared to the CRDSA, the most different aspect of the IRSA schemes lies in the multiple weights of SPs. In the CSA scheme, as a generalization of IRSA scheme, before the transmission, the information packets from each user are partitioned and encoded into multiple shorter packets via local packet-oriented codes at the media access control layer. Correspondingly, at the receiver side, the SIC process combined with the local decoding for the packet-oriented codes to recover collided packets. The construction of the SPs is one of the key challenges to increasing the throughput of these slotted ALOHA schemes. To this end, the extrinsic information transfer (EXIT) chart is used and the asymptotic performance \cite{Asymptotic_Perfo_of_CSA_CL:2018} and non-asymptotic performance \cite{Non_asymptotic_CSA_ISIT:2019} of CSA schemes were studied.

In slotted ALOHA schemes, the transmitted packets are suffered by two kinds of erasure channels, named slot erasure channel (SEC) and packet erasure channel (PEC). For wireless communication systems with limited transmit powers, the strong external interference may overwhelm all the received packets in a particular slot interval. It implies that all of the transmitted packets in a slot are erased with a certain probability, which is namely SEC. Besides, due to the effect of deep fading in wireless transmissions, there exists a certain erasure rate for the transmitted packets, namely PEC. In \cite{err_floor_CSA}, the performance of the IRSA scheme over PECs was investigated, where the error floor of packet loss rate was analyzed and the code distributions were designed to minimize this error floor. In \cite{CSAErasureChannels:Sun}, the design and analysis of CSA and IRSA schemes over erasure channels (include SECs and PECs) were investigated, and the asymptotic throughput of CSA and IRSA schemes over erasure channels were derived.

Furthermore, some practical issues should be addressed for the ISRA and CSA schemes. The first issue is pointer processing. As mentioned in \cite{CSA_magz:Paolini}, one of the underpinning assumptions for the slotted ALOHA is that each replica is equipped with pointers to the slots containing other replicas transmitted by the same user. However, when the massive active users access a slotted ALOHA scheme, it is not trivial to generate the pointers, nor is the cost of sending many pointers negligible. For short packet communication, the random access protocol sequence is used as each user's SP to avoid the pointer operation \cite{Grouptest_shortpacket_ISIT:2019}. Another more elegant approach to address this issue is to embed in each replica a user-specific seed of a pseudorandom generator or a row index of a constructed SPs' look-up table, which are known both for the users and the base station (BS) \cite{CSA_magz:Paolini}. The second issue is the SP construction for the erasure channels. To address this issue, the EXIT charts are using to optimize the SPs. Nevertheless, the asymptotic throughputs of CSA and IRSA schemes over the erasure channels show that this optimization method is a capacity-approaching method \cite{CSAErasureChannels:Sun}. Other issues of slotted ALOHA schemes were also researched in \cite{Enhanc_CRDSA_Power_Div_TVT:2018} \cite{PLNetCoding_for_RA_TVT:2019}.

Recently, as a new concept in information theory, channel polarization was discovered in the constructive capacity-achieving families of codes for symmetric memoryless channels and later generalized to source coding, multiuser channels, and other problems. The codes using the polarization phenomenon to construct (encode) is named polar codes \cite{polarcode_Arikan}, which provably achieve the capacity of any symmetric memoryless channels with successive cancellation (SC) decoding.

To address the above issues, we propose a new polar slotted ALOHA (PSA) framework, which use channel polarization to construct the identical SP set within each active user and the BS. Before transmitting, the identical SP set is constructed within each active user and BS when the number of active users is known. In this way, the handling pointers' procedure is avoided in the PSA schemes. Moreover, the asymptotic throughput of PSA schemes is capacity-achieving when the number of slots within the slot-frame approaches infinity. The contributions of this work are summarized as:
\begin{enumerate}
  \item[1)] A theoretical analysis framework for the PSA schemes over SECs is provided. Based on the packet-oriented operation for the overlap packets when they conflict in a slot, it is demonstrated that the operation guarantees the packet-based polarization transform maintains the polarization phenomenon regardless of the length of the packet. And it is proved that the capacity of SECs is achievable when the number of slots in the slot-frame tends to infinity.
  \item[2)] Two SP assignment methods for the PSA scheme are developed guided by the packet-based polarization. One is the SPA method with variable SEP (SPA-v) and another is the SPA method with a fixed SEP value (SPA-f). In the procedure of the two SPA methods, for each user and the BS, a capacity-ordered index sequence $c_1^N$ is first computed, and then, the identical SP sets are constructed. Finally, with the aid of the $c_1^N$, each user selects their own SP from the SP set. The different aspect of the two SPA methods lies in the calculating process of $c_1^N$. The SPA-v is online computing the sequence $c_1^N$ with a variable SEP. However, the SPA-f is offline calculating $c_1^N$ with a fixed SEP value and pro-stored the sequence $c_1^N$ into a look-up table and equipped in each user and the BS.
  \item[3)] A packet-oriented successive cancellation (pSC) and a pSC list (pSCL) decoding algorithm of the PSA are developed. Furthermore, the finite-slot non-asymptotic throughput bounds and the asymptotic throughput of the PSA schemes using the pSC decoding are investigated.
\end{enumerate}

The paper is organized as follows. Firstly, the slotted ALOHA system model and some preliminaries are introduced in Section \ref{section_system_model_and_some_preliminaries}.  Second, the polarization transformation for the SECs based on the packet-oriented operation is investigated in Section \ref{section Polarization Transformation for SECs based on Packet-Oriented Operation}. The PSA scheme includes two SPA methods and the pSC/SCL decoding algorithm are presented, and the throughput analysis of the PSA is also provided in Section \ref{section_Proposed_PSA_schemes_over_SECs}. Finally, simulation results are presented in Section \ref{section_simulation_results} and the conclusions are given in Section \ref{section_conclusions}.

\section{System Model and Some Preliminaries}\label{section_system_model_and_some_preliminaries}

\begin{figure*}[htbp]
\centerline{\includegraphics[scale=0.75]{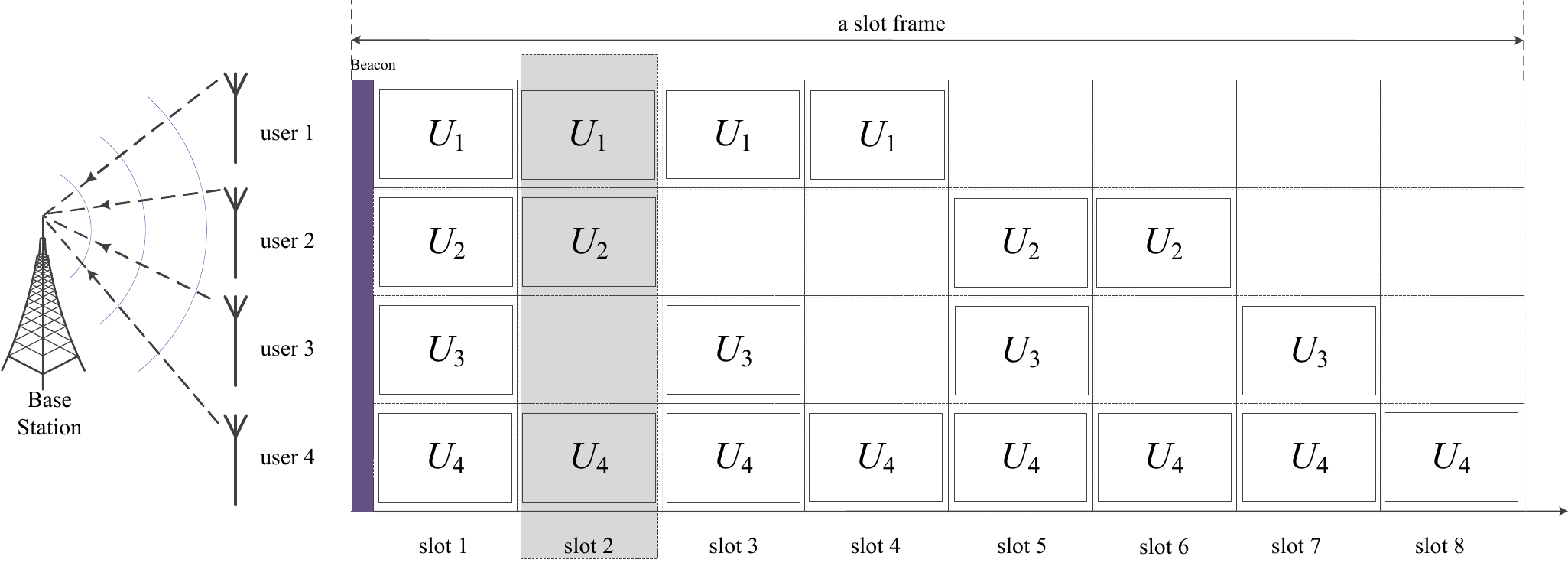}}
\caption{An example of the proposed PSA schemes over the slot erasure channels where the slotted packet within the $2$nd slot of the slot-frame is erased in the BS receiver. }
\label{fig.1}
\end{figure*}

In a slot-frame for the slotted ALOHA system, there are $M$ active users who attempt to transmit their information packets to a common receiver BS via a shared channel which consists of $N$ slots with an identical duration. We call the packets which active users transmit to the BS as information packets. At the received side, the packet in each slot of the slot-frame is named as slotted packet $X_i$, $i \in \{1,2,..,N\}$.

Similar to the previous works of slotted ALOHA schemes, we make the following assumptions:
\begin{enumerate}
  \item[A.1] Each uncoordinated user transmits a single information packet per slot-frame;
  \item[A.2] The number of active users $M$ is identified by each active user and the BS;
  \item[A.3] Each information packet or slotted packet contains $r$ bits and fits one slot interval. That is, each packet can be described as a bit-vector with $r$ elements.
\end{enumerate}

The following notations are used in the paper. We use the notation $\overline{x}\triangleq 1-x$ for $x\in \{0,1\}$. The information packet of the user $t$ is denoted by $U_t =({u}_{t,1},{u}_{t,2},...,{u}_{t,r}) = ({u}_t)_1^r  \in \{0,1\}^r$, $0\leqslant t\leqslant M$, and a slotted packet is denoted by $X_k = (x_k)_1^r = (x_{k,1},x_{k,2},...,x_{k,r})$, and hence $X_k[i]={x}_{k,i}$ is the $i$th elements of the slotted packet, $1\leqslant k\leqslant N$. Besides, an $N \times N$ matrix is denoted by ${{\bf{F}}_N}$. The ${\rm{{\cal S}}}$ denotes a set, and the cardinality of the set ${\rm{\cal S}}$ is denoted by $|{\rm{{\cal S}}}|$.

\subsection{Slotted ALOHA Procedure}

An example of the slotted ALOHA schemes over an SEC is shown in Fig. \ref{fig.1}. There are $M = 4$ active users who want to transmit information packets $U_1^4$ to the BS by through the slotted ALOHA scheme which includes $N=8$ slots in each slot-frame. Before received by the BS, the packets are suffered by the SEC, and resulting in some slotted packets are erased. As shown in Fig. \ref{fig.1}, the second slotted packet is erased which is indicated as a gray slot interval.

Just like playing a carousel game in a playground, we need two phases, waiting and playing. In addition to the packet diversity due to the copies, there is a waiting procedure before the active users access the slotted ALOHA phase which is controlled by BS beacon \cite{beacon} labeled in Fig. \ref{fig.1}.

While waiting for the random access a slotted ALOHA scheme, we assume that each active user broadcasts an access request signal in a random time and the access request signal can be detected by other active users and the BS. Following this assumption, the number $M$ of the active users can be identified by each active user and the BS by using the statistical counting method. \footnote{The collision of access request signals should be avoided when the request is initiated. It is assumed that each user can detect whether their request conflicts. When their request conflicts, they are asked to withdraw and wait for a random time to initiate an access request again.}

When each active user received the BS beacon and before transmitting their information packets, we assume the following assumption holds:
\begin{enumerate}
  \item[A.4] Each active user and the BS can identify the order of active users by detecting the request queue during the waiting process. The order of active users is indicated by labeling user $M$, ... , user $2$, user $1$.
\end{enumerate}

The user $M$ is an active user whose request is first detected, ..., and so on. That is, the input of the slotted ALOHA is an ordered information packet sequence $(U_1,...,U_M)$.

Similar to previous works on slotted ALOHA schemes, the offered traffic load (packets/slot) of PSA is defined as
\begin{equation}\label{eq_offered_traffic}
  G = M/N.
\end{equation}

The throughput efficiency (packets/slot) is defined as
\begin{equation}\label{eq_throughput_efficiency}
  T = GP_u
\end{equation}
where the $P_u$ is the information packets recovery probability of all active users in the BS receiver.

\subsection{Polar Codes}

Polar codes are a new class of error-correcting codes, proposed by Ar{\i}kan in \cite{polarcode_Arikan}, which provably achieve the capacity of any symmetric binary-input memoryless channels with an efficient SC decoding. The asymptotic effectiveness of SC decoding derives from the fact that the polarized synthetic channels tend to become either noiseless or completely noisy, as the block-length goes to infinity. In the polar encoding process, the noiseless polarized channels are used to send the information bits, and the rest polarized channels are assigned by the fixed values, such as zeros. Mixing information bits with fixed bits to form a source bit sequence $u_1^N$.

The main process of polar transformation is to combine the source bit sequence $u_1^N$ by repeated applying the polarization kernel
$ {\bf{F}}_2 =
  \big[
    \begin{smallmatrix}
        1 &0 \\
        1 &1
    \end{smallmatrix}
\big] $ with $n$ times, and hence $N=2^n$, $n \in \{1,2,...\}$. The generator matrix is formed by selecting the row vectors of ${\bf{F}}_2^{ \otimes n}$ with indices within the information index set $\mathcal{I}$ , where ${\left( \cdot  \right)^{ \otimes n}}$ denotes the $n$th Kronecker power. One of the key challenges for the polar codes is to construct the information index set $\mathcal{I}$ which is governed by the reliability metrics of the polarized channels.

Mathematically, the encoded bit sequence is
\begin{equation}\label{Eq_Polar_encoding}
  x_1^N = u_1^N \cdot {\bf{F}}_2^{ \otimes n}
\end{equation}
where the $|\mathcal{I}|$ information bits are loaded into the source sequence $u_1^{N}$ at position indices within index set $\mathcal{I}$, and other bits are set to the fixed zero values \cite{polarcode_Arikan}.

\section{Polarization Transformation for SECs based on Packet-Oriented Operation}\label{section Polarization Transformation for SECs based on Packet-Oriented Operation}
In this section, we first show the packet-oriented operation for the packets when they are overlapped within a slot of the slotted ALOHA frame. And then, the SEC and its equivalent product compound channel model are analyzed. Finally, packet-based polarization transformation for $2^r$-ary SECs based on the packet-oriented operation is investigated.

\subsection{Packet-Oriented Operation for Overlap Packets}
The packet-oriented operation for the overlap packets, when they conflict in a slot, has the following properties:
\begin{enumerate}
  \item Each information packet and the overlapped packet are fit in exactly one slot interval. That is, under the assumption A.$3$, the length of bits in the overlapped packet equals to that of the input information packet, which means that the output packet of the operation for the overlapped packets have a bit-width consistent with each input information packet;
  \item The packet-oriented operation is reversible. That is to say, assuming that two information packets are operated by the packet-oriented operation, when one of two information packets is clean, another information packet is completely recovered by the reverse operation.
\end{enumerate}

\begin{figure}[htbp]
\centerline{\includegraphics[scale=0.75]{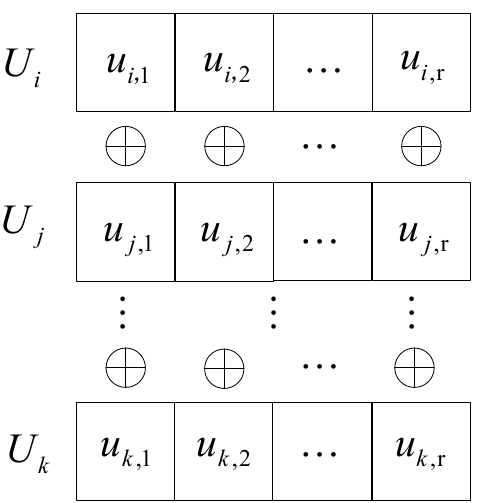}}
\caption{The packet-oriented operation for the packets when they conflict in a slot.}
\label{fig.2}
\end{figure}

An illustration of a packet-oriented operation that satisfies the above two conditions is shown in Fig. \ref{fig.2}. The packet-oriented operation is based on the packet-component $mod~2$ sum and performed bit-by-bit in parallel. And then, each bit of the overlapped slotted packet is computed independently of each other without carry. Consequently, in the following, we only consider the packet-component $mod~2$ sum as the packet-oriented operation in the proposed PSA schemes.

\begin{definition}
Let vector $U_i = ({u}_{i,1},{u}_{i,2},...,{u}_{i,r})$ and vector $U_j = ({u}_{j,1},{u}_{j,2},...,{u}_{j,r})$ denote two information packets, the output of the packet-oriented operation is defined as
\begin{equation}\label{eq_def_packet_oriented_operation}
  {U_i} \boxplus {U_j} = ({u}_{i,1}\oplus {u}_{j,1},{u}_{i,2}\oplus {u}_{j,2},...,{u}_{i,r}\oplus {u}_{j,r}).
\end{equation}
\end{definition}

With the definition of the packet-oriented operation, the following propositions hold.
\begin{proposition}\label{propo_1}
Let $X_k = {U_i} \boxplus {U_j}$ and one of input vector ${U_i}$ is known, another input is clear by $U_j = {U_i} \boxplus {X_k} = {X_k} \boxplus {U_i}$.
\end{proposition}
\begin{proof}
For any $1 \leqslant h \leqslant r$, when the $X_{k}[h]$ and $U_{i}[h]$ are known, using the definition of packet-oriented operation, $X_{k}[h] = U_{i}[h] \oplus U_{j}[h]$. Obviously, $U_{j}[h] = U_{i}[h] \oplus X_{k}[h] = X_{k}[h] \oplus U_{i}[h]$. Therefore, the formula $U_{j} = U_{i} \boxplus X_{k} = X_{k} \boxplus U_{i}$ holds.
\end{proof}

\begin{proposition}\label{propo_2}
Under the packet-oriented operation $\boxplus$ with two input vectors $U_i$ and $U_j$, each element of the output vector $X_{k}[h] = U_{i}[h] \oplus U_{j}[h], 1\leqslant h \leqslant r$, is only dependent on the $h$th element of the input vectors.
\end{proposition}

\begin{proof}
Without loss of generality, supposed  $X_{k}[h] = U_{i}[h]\oplus U_{j}[t]$, $h\neq t$, $1\leqslant h,t \leqslant r$. Using Proposition \ref{propo_1}, we get $U_{j}[t] = U_{i}[h] \oplus X_{k}[h] = U_{j}[h]$, and which means that $t=h$ holds, so the original hypothesis is not true. Therefore, this proposition is established.
\end{proof}

\subsection{Slot Erasure Channel}

Let $S(\epsilon):\mathcal{X}^r \rightarrow \mathcal{Y}^r$, $\mathcal{X}^r = \{0,1\}^r$, $\mathcal{Y}^r = \{0,1\}^r \cup \{E\}$ be an SEC with input alphabet $\mathcal{X}^r$, output alphabet $\mathcal{Y}^r$, and transition probabilities $S({{y}}^r|{x}^r), {y}^r \in {\mathcal{Y}^r}, {{x}}^r \in {\mathcal{X}^r}$, and $E=(e,e,...,e)$ is the erasure packet of SEC with the probability $\epsilon$. Obviously, the SEC $S(\epsilon)$ is a symmetric channel, and $|\mathcal{X}^r|=2^r \triangleq q$, $|\mathcal{Y}^r|=2^r+1 = q+1$. When $r=1$, the SEC $S(\epsilon)$ degenerates into a binary erasure channel (BEC) $W(\epsilon): \mathcal{X} \rightarrow \mathcal{Y}$, $\mathcal{X} = \{0,1\}$, $\mathcal{Y} = \{0,e,1\}$.

The channel capacity of $S$ and $W$ is denoted by $I(S)$ and $I(W)$. Obviously, the capacity of a BEC $W$ with erasure probability $\epsilon$ is $I(W) = 1-\epsilon$.

\begin{lemma} \label{capacity_SEC}
The symmetric capacity of the SEC $S(\epsilon)$ is $r(1-\epsilon)$ bits/channel use.
\end{lemma}

\begin{proof}
The transition probabilities of SEC $S(\epsilon)$ are
\begin{equation}
  S({y}^r|{x}^r)= \left\{ \begin{array}{ll}
                             1-\epsilon & \textrm{if ${y}^r = {x}^r $}\\
                               \epsilon & \textrm{if ${y}^r = {E}$}
                           \end{array} \right. \label{eq_solt_erasure_model}
\end{equation}

With the probability of each input alphabet $1/q$, the symmetric capacity of the channel $S$ can be written as
\begin{equation}
   I(S) = \sum_{x^r \in {\mathcal{X}^r} }\sum_{y^r \in {\mathcal{Y}^r} } {\frac{1}{q} S({y}^r|{x}^r) \log \frac{S({y}^r|{x}^r)}{\sum_{(x')^r \in {\mathcal{X}^r} }{\frac{1}{q} S({y}^r|(x')^r)} } }. \label{eq_def_sym_capacity}
\end{equation}

Substituting Eq. (\ref{eq_solt_erasure_model}) into Eq. (\ref{eq_def_sym_capacity}), we obtain that
\begin{equation*}
   \begin{aligned}
       I(S) & =  \sum_{x^r \in {\mathcal{X}^r} }\sum_{y^r \in {\mathcal{Y}^r} } {\frac{1}{q} S({y}^r|{x}^r) \log \frac{S({y}^r|{x}^r)}{ \sum_{(x')^r \in {\mathcal{X}^r} }{\frac{1}{q} S({y}^r|(x')^r)} } }  \\
            & + \frac{1}{q} \sum_{x^r \in {\mathcal{X}^r} } \sum_{y^r = {{E}} }  {S({y}^r|{x}^r) \log \frac{ q S({y}^r|{x}^r)}{ \sum_{(x')^r \in {\mathcal{X}^r} }{ S({y}^r|(x')^r)} } } \\
            & = \frac{1}{q} \sum_{{y^r=x^r};{x^r} \in {\mathcal{X}^r} } { S({y}^r|{x}^r) \log \frac{ {q} S({y}^r|{x}^r)}{ \sum_{(x')^r \in {\mathcal{X}^r} }{S({y}^r|(x')^r } } } \\
            & + \frac{1}{q} \sum_{x^r \in {\mathcal{X}^r} }  { S(E|{x}^r) \log \frac{ {q} S(E|{x}^r)}{ \sum_{(x')^r \in {\mathcal{X}^r} }{S(E|(x')^r } } } \\
            & = \frac{1}{q} \Bigg{(} q { (1-\epsilon) \log \frac{ {q} (1-\epsilon) }{ (1-\epsilon) } } \Bigg{)}
                + \frac{1}{q} \Bigg{(} q \times { \epsilon \times \log \frac{ {q}\times\epsilon }{ q\times\epsilon } } \Bigg{)}\\
            & =  r(1-\epsilon),
      \end{aligned}
\end{equation*}
where the base of the logarithm is $2$, and then the unit is bits/channel use.
\end{proof}

\begin{remark}\label{remark_relationship_SEC_BEC}
From \emph{Lemma} \ref{capacity_SEC}, we get that the capacity of a SEC $S(\epsilon)$ is $r$ times that of a BEC $W(\epsilon)$.
\end{remark}

From \emph{Remark} \ref{remark_relationship_SEC_BEC} and the definition of the product channel \cite{product_channel_Shannon:1956}\cite{parallel channels:1976}, the SEC $S$ can be expressed as a product compound channel which contains $r$ identical BECs, and it can be shown that
\begin{equation}\label{SEC}
  S = \underbrace{W \times W... \times W}_{r}
\end{equation}
where the $r$ identical BECs $W$ mean that they suffer the identical channel noise realization with that of the SEC $S$, and this is caused by the packet-oriented operation when the slotted packets over the SEC.

\subsection{Polarization Transformation for SECs}

\begin{figure}[htbp]
\centerline{\includegraphics[scale=1]{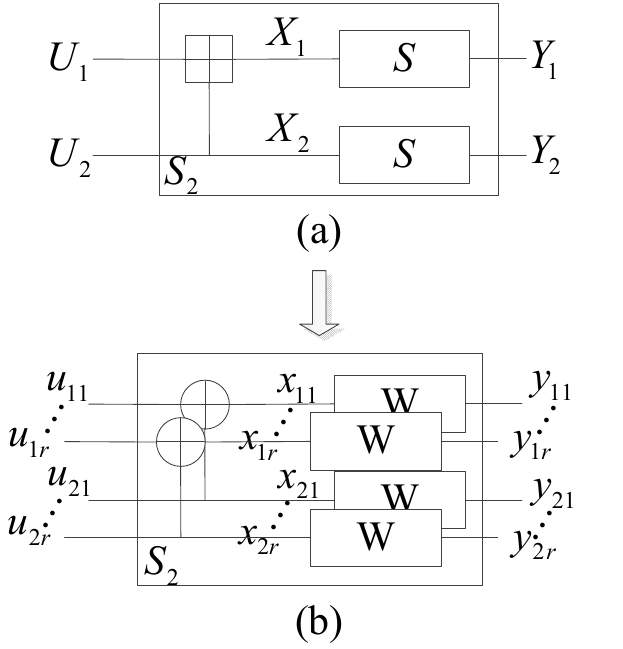}}
\caption{(a) Using the Ar{\i}kan polarizing Kernel ${\mathbf{F}}_2$ with packet-oriented operation $\boxplus$, the combined channel $S_2$ of the SEC $S$. (b) The equivalent of the combined channel $S_2$ comprises $r$ identical combined channels $W_2$.}
\label{fig.3}
\end{figure}

For the SECs, the packet-oriented combining channel can be also generated by recursive using the Ar{\i}kan's polarizing kernel ${\mathbf{F}}_2$ with the packet-oriented operation $\boxplus$. For the first level of the recursion combines two independent copies of SEC $S$ as shown in Fig. \ref{fig.3}.(a) and obtains the channel $S_2$: $(\mathcal{X}^r)^2 \rightarrow (\mathcal{Y}^r)^2$ with the transition probabilities
\begin{equation}
   S_{2}(Y_1,Y_2|U_1,U_2) \triangleq S(Y_1|U_1 \boxplus U_2)S(Y_2|U_2). \label{eq_def_S2}
\end{equation}

The synthetic channels $S^{-}$ and $S^{+}$ of the combined channel $S_2$ are defined as
\begin{equation}
   S^{+}(Y_1,Y_2,U_1|U_2) \triangleq {\frac{1}{q} S_2(Y_1,Y_2|U_1,U_2)} \label{eq_def_S+}
\end{equation}
and
\begin{equation}
   S^{-}(Y_1,Y_2|U_1) \triangleq \sum_{U_2 \in \mathcal{X}^r } {\frac{1}{q} S_2(Y_1,Y_2|U_1,U_2)}. \label{eq_def_S-}
\end{equation}
where $U_1,U_2 \in \mathcal{X}^r$, $Y_1,Y_2 \in \mathcal{Y}^r$.

\begin{lemma}\label{lemma_SEC_product_compound_channel}
The combined channel $S_2$ under the packet-oriented operation $\boxplus$ can be viewed as a product compound channel with $r$ identical combined channels $W_2$ with the $mod~2$ operation. That is

\begin{equation}\label{Eq_S2_Equivalent}
   S_{2}(Y_1,Y_2|U_1,U_2) = \prod_{i=1}^{r} {W_{2}(y_{1,i},y_{2,i}|u_{1,i},u_{2,i})}
\end{equation}
and its two synthetic channels are

\begin{equation}\label{Eq_S+_Equivalent}
   S^{+}(Y_1,Y_2,U_1|U_2) = \prod_{i=1}^{r} {W^{+}(y_{1,i},y_{2,i},u_{1,i}|u_{2,i})}
\end{equation}

\begin{equation}\label{Eq_S-_Equivalent}
   S^{-}(Y_1,Y_2|U_1) =  \prod_{i=1}^{r}{W^{-}(y_{1,i},y_{2,i}|u_{1,i})}.
\end{equation}

\end{lemma}

\emph{Lemma} $2$ is proved in Appendix A. From this \emph{Lemma}, the equivalent of the combined channel $S_2$ is shown in Fig. \ref{fig.3}.(b), which is a compound channel with $r$ identical of a combined channels $W_2$.

\begin{remark}\label{remark_SEC_product_compound_channel}
This transformation can be applied recursively to the two synthetic channels $S^{+}$, $S^{-}$ resulting in four synthetic channels of the form $S^{t_{1}t_{2}}$,$t_1,t_2 \in \{ +,- \}$. After $n$ steps, we obtain $N=2^n$ synthetic channels $S^{(j)}_{N}(Y_{1}^{N},U^{j-1}_{1}|U_{j})$\footnote{There is a bijection mapping between the left-most-significant-bit binary representation $j$ and vector $t_1^{n}\in \{+,-\}^n$ by replacing each $-$ that appears in $t_1^n$ with $0$ and each $+$ that appears in $t_1^n$ with a $1$.}, $j = 1,...,N$, and
\begin{equation}\label{Eq_relation_S_and_W}
   S^{(j)}_{N} = \prod_{i=1}^{r} {W^{(j)}_{N}(y_{1,i},...,y_{N,i},u_{1,i},...,u_{j-1,i}|u_{j,i})}
\end{equation}
which shows that the combined channel $S_{N}$ is equivalent to a product compound channel with $r$ identical combined channels $W_N$.
\end{remark}

On the capacity and Bhattacharyya parameter, a relationship between the compound SEC synthetic channels and its component BEC synthetic channels have the following \emph{Lemma} \ref{lemma_capcaity_combined_polarized_channel}.

\begin{lemma}\label{lemma_capcaity_combined_polarized_channel}
When a slotted ALOHA scheme suffered by an SEC with the SEP $\epsilon$, for $1\leqslant j\leqslant N$, the capacity and the Bhattacharyya parameter of the SEC synthetic channels are
\begin{equation} \label{eq_capcaity_combined_polarized_channel}
   I\Big{(}S_N^{(j)}\Big{)} = r I\Big{(}W_N^{(j)}\Big{)}
\end{equation}
\begin{equation} \label{eq_Z_parameter_combined_polarized_channel}
   Z\Big{(}S_N^{(j)}\Big{)} = r Z\Big{(}W_N^{(j)}\Big{)}.
\end{equation}

\end{lemma}

\emph{Lemma} \ref{lemma_capcaity_combined_polarized_channel} is proved in Appendix B. With the conclusion of \emph{Lemma} \ref{lemma_capcaity_combined_polarized_channel}, the capacity and Bhattacharyya parameter recursion formula of the SEC synthetic channels can be obtained as they are shown in \emph{Theorem} \ref{theorem_SEC_capacity_and_Z_parameter}.

\begin{theorem}\label{theorem_SEC_capacity_and_Z_parameter}
For $1\leqslant j \leqslant N/2$ and $N=2^n$, the capacity of SEC synthetic channels can be recursively calculated as
\begin{equation}\label{SEC_info_capacity}
  \left\{\begin{aligned}
   I\Big{(}S_N^{(2j-1)}\Big{)} & = \frac{1}{r}I\Big{(}S_{N/2}^{(j)}\Big{)}^2 \\
     I\Big{(}S_N^{(2j)}\Big{)} & = 2I\Big{(}S_{N/2}^{(j)}\Big{)} - \frac{1}{r}I\Big{(}S_{N/2}^{(j)}\Big{)}^2
   \end{aligned}\right.
\end{equation}
and the Bhattacharyya parameter for SEC synthetic channels are
\begin{equation}\label{SEC_Z_parameter}
  \left\{ \begin{aligned}
 Z\left( {S_{N}^{\left( {2j - 1} \right)}} \right) &= 2Z\left( {S_{N/2}^{\left( j \right)}} \right) - \frac{1}{r}{Z^2}\left( {S_{N/2}^{\left( j \right)}} \right) \\
 Z\left( {S_{N}^{\left( {2j} \right)}} \right) &= \frac{1}{r}{Z^2}\left( {S_{N/2}^{\left( j \right)}} \right). \\
 \end{aligned} \right.
\end{equation}
\end{theorem}

The proof of \emph{Theorem} \ref{theorem_SEC_capacity_and_Z_parameter} is given in Appendix C. The above capacity parameters $I(S_N^{(j)})$ and Bhattacharyya parameters $Z(S_N^{(j)})$ are two metrics of the rate and reliability (with respect to bits (packet elements) ) of SEC synthetic channels $S_N^{(j)}$, $1\leqslant j\leqslant N$. Subsequently, we will investigate the polarization phenomenon of SEC synthetic channels.

In the case $r=1$, the SEC synthetic channels degenerate into the BEC synthetic channels and are denoted by $W^{(j)}_{N}, 1\leqslant j\leqslant N$, $N=2^n$. In \cite{polarcode_Arikan}, it is proved that as $N$ increases, the synthetic channels $W^{(j)}_{N}$ become either almost perfect or almost completely noisy (polarize). It means that, in formal terms, for any $\gamma > 0$, the following formula holds
\begin{equation}
   \lim_{n\rightarrow \infty} {\frac{|t\in \{ +,- \}^{n}: I\Big{(}W_N^{(t)}\Big{)} \in (\gamma,1-\gamma)|}{2^n}} = 0. \label{eq_bit_polarization_effects}
\end{equation}

When $r>1$, for the SEC synthetic channels, there is a similar polarization phenomenon that is described as \emph{Theorem} $2$.

\begin{theorem}\label{theorem_SEC_polarization_effect}
As $N$ increases, the channels $S^{(j)}_{N}$ become either almost perfect or almost completely noisy in the symmetric SECs. That is, for any $\gamma > 0$,
\begin{equation}
   \lim_{n\rightarrow \infty} {\frac{|t\in \{ +,- \}^{n}: I\Big{(}S_N^{(t)}\Big{)} \in (\gamma,1-\gamma)|}{2^n}} = 0. 
\end{equation}
\end{theorem}

\begin{proof}
From \emph{Lemma} \ref{lemma_capcaity_combined_polarized_channel}, each capacity of synthetic channels $S_N^{(t)}$ can be computed as $I\Big{(}S_N^{(t)}\Big{)} = r I\Big{(}W_N^{(t)}\Big{)}$, $1\leqslant t\leqslant N$, $N=2^n$. For any $\gamma > 0$, and let $v=\gamma/r$, we have
\begin{equation*}
\begin{aligned}
  &\lim_{n\rightarrow \infty} {\frac{|t\in \{ +,- \}^{n}: I\Big{(}S_N^{(t)}\Big{)} \in (\gamma,r-\gamma)|}{2^n}} \\
= &\lim_{n\rightarrow \infty} {\frac{|t\in \{ +,- \}^{n}: r I\Big{(}W_N^{(t)}\Big{)} \in (\gamma,r-\gamma)|}{2^n}}\\
= &\lim_{n\rightarrow \infty} {\frac{|t\in \{ +,- \}^{n}: I\Big{(}W_N^{(t)}\Big{)} \in (\gamma/r,1-\gamma/r)|}{2^n}}\\
= &\lim_{n\rightarrow \infty} {\frac{|t\in \{ +,- \}^{n}: I\Big{(}W_N^{(t)}\Big{)} \in (v,1-v)|}{2^n}}\\
= &~0.
\end{aligned}
\end{equation*}
The proof of \emph{Theorem} \ref{theorem_SEC_polarization_effect} completes.
\end{proof}

From \emph{Theorem} \ref{theorem_SEC_polarization_effect}, as the $N$ increases, the capacity of some channels $S_N^{(i)}, i\in \mathcal{I}$, tend to $r$ (bits/channel use), and that of the rest channels tend to $0$, which is the polarization phenomenon of SEC synthetic channels $S_N^{(i)}$. In the proposed PSA schemes, we use the index set $\mathcal{I}$ as the information packet position index set that indicates which row vectors of the matrix ${\mathbf{F}}_2^{ \otimes n}$ are selected into the SP set.

\section{Proposed PSA schemes over SECs}\label{section_Proposed_PSA_schemes_over_SECs}

In this section, the proposed PSA schemes included two SPA methods and the pSC/SCL decoding algorithms are investigated in detail. Finally, the finite-slots non-asymptotic throughput bounds and the asymptotic throughput for the PSA scheme using the pSC decoding are analyzed.

\subsection{SP Assignment Methods}

With the conclusion from section \ref{section Polarization Transformation for SECs based on Packet-Oriented Operation} and guided by the polar encoding, there are two SPA methods for the PSA scheme. One is the SPA method with a variable SEP (SPA-v) and another is the SPA method with a fixed SEP value (SPA-f). The detailed procedure of the SPA-v algorithm and the SPA-f algorithm are described in Algorithm \ref{alg:SPA-v-algorithm} and Algorithm \ref{alg:SPA-f-algorithm}.

\begin{algorithm}
\caption{SPA-v algorithm for the PSA scheme}
\label{alg:SPA-v-algorithm}
\begin{algorithmic}
\REQUIRE{$U_1^{M}$: an ordered information packet sequence\;
~~~~~~~~~~~$\epsilon$: the variable SEP. } 
\ENSURE{SP vectors: $V_1,V_2,...,V_M $\;
~~~~~~~~$c_1^{N}$: capacity-ordered index sequence\;
~~~~~~~~~~$\mathcal{I}$: the information packet index set.}       
\STATE $\mathbf{Step~1}$: (online) Using the SEP $\epsilon$, the capacity metric $I\Big{(}S_N^{(i)}\Big{)}$ of synthetic channels are online calculated by using Eq. (\ref{SEC_info_capacity}), $1 \leqslant i \leqslant N$\;
\STATE $\mathbf{Step~2}$: (online) Sorting the values $I\Big{(}S_N^{(i)}\Big{)}$, $1$$ \leqslant i \leqslant $$N$, then the capacity-ordered index sequence $c_1^N$ is obtained. That is, the relationship $I\Big{(}S_N^{(c_N)}\Big{)} \leqslant \cdots \leqslant I\Big{(}S_N^{(c_1)}\Big{)}$ holds. In each active user and the BS, the indices $c_1$, $c_2$, $...$ and $c_M$ are selected into information packet position index set $\rm{{\cal I}}$.
\STATE $\mathbf{Step~3}$: SP assignment for the information packets $U_1^M$: the user $M$ selects the $c_1$th row-vector of ${\mathbf{F}}_2^{ \otimes n}$ as $V_M$ for the information packet $U_M$, ..., and the user $1$ selects the $c_M$th row-vector of ${\mathbf{F}}_2^{ \otimes n}$ as $V_1$ for the packet $U_1$.

\STATE $\mathbf{return}$: $V_1,V_2,...,V_M$, $c_1^N$ and $\mathcal{I}$.
\end{algorithmic}
\end{algorithm}

The first step of the SPA-v algorithm is online recursive computing the capacity metric of synthetic channels $I(S_N^{(i)})$ by using the Eq. (\ref{SEC_info_capacity}) with the initial values $I\Big{(}(S)_1^{(1)}\Big{)} = r(1-\epsilon)$, $1\leqslant i\leqslant N$.

The second step of the SPA-v algorithm is sorting the capacity metric of each synthetic channel $I\Big{(}S_N^{(i)}\Big{)}$, $1\leqslant i\leqslant N$. That is, the capacity-ordered index sequence $c_1^{N}$ of the synthetic channels is obtained, such that $I\Big{(}S_N^{(c_N)}\Big{)} \leqslant I\Big{(}S_N^{(c_{N-1})}\Big{)} \cdots \leqslant I\Big{(}S_N^{(c_1)}\Big{)}$. And then, with the number $M$ of the active users which is obtained before transmitting their information packets, the $M$ indices of the bigger values in the synthetic channels capacity metric sequence, $c_1,c_2,...,c_M$, are selected to constitute the information packet index set $\rm{{\cal I}}$ in each active user and the BS.

In the third step of the SPA-v algorithm, SPs are assigned to each active user. Under the assumption A.$4$, the user $i$ selects the $c_{(M-i+1)}$th row of the matrix ${\mathbf{F}}_2^{ \otimes n}$ as its SP for their information packet $U_i$, $1\leqslant i\leqslant M$.

$\mathbf{Example~~1:}$
As shown in Fig. \ref{fig.1}, there are $M=4$ active users who want to transmit information packets $U_1^4$ to the BS by through the slotted ALOHA scheme which includes $N=8$ slots in each slot-frame. With a SEP value $\epsilon = 0.5$, following by the computing and sorting steps, the capacity-ordered index sequence $c_1^8=(8,7,6,4,5,3,2,1)$ is obtained. So, the index set of information packets is $\mathcal{I}=\{8,7,6,4\}$. That is to say, the constructed SP set includes the $8$th, $7$th, $6$th and $4$th row of the ${\mathbf{F}}_2^{ \otimes 4}$. Finally, the user $4$ select the $8$th row of ${\mathbf{F}}_2^{ \otimes 4}$ with the biggest value capacity as its SP $V_1=(1,1,1,1,1,1,1,1)$, ..., and user $1$ select the $4$th row as its SP $V_4 = (1,1,1,1,0,0,0,0)$. The matrix ${\mathbf{F}}_2^{ \otimes 4}$ is shown as
\begin{equation*}
  {\mathbf{F}}_2^{ \otimes 4}=\left[
   \begin{array}{cccccccc}
   1 & 0 & 0 & 0 & 0 & 0 & 0 & 0 \\
   1 & 1 & 0 & 0 & 0 & 0 & 0 & 0 \\
   1 & 0 & 1 & 0 & 0 & 0 & 0 & 0 \\
   1 & 1 & 1 & 1 & 0 & 0 & 0 & 0 \\
   1 & 0 & 0 & 0 & 1 & 0 & 0 & 0 \\
   1 & 1 & 0 & 0 & 1 & 1 & 0 & 0 \\
   1 & 0 & 1 & 0 & 1 & 0 & 1 & 0 \\
   1 & 1 & 1 & 1 & 1 & 1 & 1 & 1 \\
   \end{array}
\right].
\end{equation*}

\begin{algorithm}
\caption{SPA-f algorithm for the PSA scheme}
\label{alg:SPA-f-algorithm}
\begin{algorithmic}
\REQUIRE{$U_1^{M}$: an ordered information packet sequence;}
\ENSURE{SP vectors: $V_1,V_2,...,V_M $\;
~~~~~~~~$c_1^{N}$: capacity-ordered index sequence\;
~~~~~~~~~~$\mathcal{I}$: the information packet index set.}       
\STATE $\mathbf{Step~1}$: (offline) The capacity metric $I(S_N^{(i)})$ of synthetic channels were offline computed by using Eq. (\ref{SEC_info_capacity}) with a fixed SEP value, $1 \leqslant i \leqslant N$. And then, the sequence $c_1^N$ was obtained by sorting the values $I\Big{(}S_N^{(i)}\Big{)}$, $1 \leqslant i \leqslant N$. That is, the relationship $I\Big{(}S_N^{(c_N)}\Big{)} \leqslant \cdots$$\leqslant I\Big{(}S_N^{(c_1)}\Big{)}$ holds. In each active user and the BS, the indices $c_1$, $c_2$, $...$ and $c_M$ are selected into information packet position index set $\rm{{\cal I}}$.
\STATE $\mathbf{Step~2}$: The sequence $c_1^N$ was pro-stored into a look-up table and equipped in each user and the BS\;
\STATE $\mathbf{Step~3}$: SP assignment for the information packets $U_1^M$: the user $M$ selects the $c_1$th row-vector of ${\mathbf{F}}_2^{ \otimes n}$ as $V_M$ for the information packet $U_M$, ..., and the user $1$ selects the $c_M$th row-vector of ${\mathbf{F}}_2^{ \otimes n}$ as $V_1$ for the packet $U_1$.

\STATE $\mathbf{return}$: $V_1,V_2,...,V_M $, $c_1^N$ and $\mathcal{I}$.
\end{algorithmic}
\end{algorithm}

In the SPA-v algorithm, the SP set for each slot-frame is online constructed by using the variable SEP value of SEC in each active user and the BS simultaneously. For reducing the user's computational complexity, in the SPA-f algorithm, a fixed capacity-ordered SP sequence is pre-stored as a lookup table and was equipped in each user and the BS. The SP sequence is offline constructed by using a fixed SEP value (It is emphasized that any value which about statistics of $\epsilon$ was known is allowed used). Compared to the SPA-v algorithm, the throughput of the PSA scheme with the SPA-f algorithm suffers from a certain throughput loss as its SP set construction method uses a fixed SEP value regardless of the variable SEP value of SECs.

Mathematically, an equivalent source packet sequence $U_1^N$ is obtained after the SP vectors are allocated. Without misunderstanding, for $1\leqslant i\leqslant N$,
\begin{equation} \label{eq_mixed_equivalent source_N}
  U_i = \left\{
  \begin{aligned}
        ~U_{j}~~~~& \textrm{if $i \in \mathcal{I}$ and $c_{(M-j+1)}=i$}\\
        \mathbf{0}~~~~& \textrm{if $i \not \in \mathcal{I}$ }
  \end{aligned}
  \right.
\end{equation}
where $c_1^N$ is the capacity-ordered index sequence. It should be noted that the equivalent source packet sequence is obtained by mixing the information packet sequence $U_1^M$ and $(N-M)$ all-zero packets. Just like the SC/SCL decoding for polar codes, in the BS, we can use the prior information about the ($N-M$) all-zero packets to aid for the decoding of the estimated source packet sequence $\hat{U}_1^N$ by using the pSC/PSCL decoding.

\subsection{Decoding Algorithms}


Before describing the pSC decoding algorithms, we define an operation $\star$ and an indicator function $f(\cdot)$ about the packets $Y \in \{0,1\}^r \cap \{E\}$ which will be used in the pSC/SCL decoding. The packet-oriented operation $\star$ is defined as:
\begin{equation}\label{eq_p_operation_with_erasure_packet}
  Y_1 \star Y_2 =
  \left\{
    \begin{aligned}
              &Y_1 {\boxplus} Y_2 ~~~\text{if}~Y_1,Y_2 \in \{0,1\}^r\\
              &E               ~~~~~~~~~\text{if}~Y_1\in\{E\}~\text{or}~Y_2\in\{E\}
\end{aligned} \right.
\end{equation}
That is, for any two packets without erasure packet, the operation $\star $ is equivalent to the operation ${\boxplus}$ which is defined as Eq. (\ref{eq_def_packet_oriented_operation}). Otherwise, if anyone of inputs is an erasure packet, the output packet of the operation $\star$ is $E$. And let $f(v)$ denote the indicator of the packet-component being or not an erasure $e$, that is,
\begin{equation}\label{Eq_indicator_func}
  f(v) \triangleq  \left\{
    \begin{aligned}
              &1~~~~~~~~~~\text{if}~~v \neq e \\
              &0~~~~~~~~~~\text{if}~~v = e
    \end{aligned} \right.
\end{equation}

Guided by the decoding based on the multi-layer graphical representation of polar codes \cite{polarcode_Arikan} \cite{polar_decoding_multi_layer:2013}, the unform graph is used for the pSC decoding of the PSA scheme. For a PSA scheme with $N=2^n$ slots, there are $N$ rows and $n+1$ columns in the associated graph. For each $1 \leqslant i \leqslant N$ and $0 \leqslant j \leqslant n$, the node in the $i$th row and the $j$th column is associated with two variables: a posterior variable $Q_{i,j}$ and an estimated variable $\hat{U}_{i,j}$. The right-most posterior variables ($Q_{i,n}: i\in\{1,...,N\}$) are the received from the slot erasure channel and constitute the pSC decoding input.

The remaining posterior values are recursively calculated as \cite{polar_DE_mori:2009}\cite{polarcode_BEC_decoding:2018}:
\begin{equation} \label{eq_decoding_posterior_values}
  Q_{i,j} = \left\{
    \begin{aligned}
      Q_{2i-1,j+1} \star Q_{2i,j+1}~~~~~~~&~\text{if}~i\leqslant {N/2}  \\
      g\big{(}Q_{2i-1,j+1}, Q_{2i, j+1},\hat{U}_{i-N/2,j}\big{)} &~\text{if}~i > {N/2}
\end{aligned} \right.
\end{equation}
where the function $g$ relies on the estimated packet value of $(\hat{U}_{i-N/2,j})$. Using the method as shown in \cite{polarcode_BEC_decoding:2018}, each element for the output packet of $g$ can be computed as follows.

\begin{figure*}[ht]
\begin{equation}\label{eq_decoding_posterior metric_g_0}
     Q_{i,j}[w] = \left\{
          \begin{aligned}
              &e,~~~~~~~~~~~~~~~~~~~~~~~~~~~~~~~~~\text{if}~f\Big{(}Q_{2i-1,j+1}[w]\Big{)}=f\Big{(}Q_{2i,j+1}[w]\Big{)} = 0~\text{or}~Q_{2i-1,j+1}[w] = \overline{{Q}_{2i,j+1}[w]}\\
              &f\Big{(}Q_{2i-1,j+1}[w]\Big{)}Q_{2i-1,j+1}[w]+f\Big{(}Q_{2i,j+1}[w]\Big{)}Q_{2i,j+1}[w],~\text{if}~f\Big{(}Q_{2i-1,j+1}[w]\Big{)}=\overline{f\Big{(}Q_{2i,j+1}[w]\Big{)}}  \\
              &Q_{2i-1,j+1}[w] Q_{2i,j+1}[w]+\overline{\overline{Q_{2i-1,j+1}[w]}~\overline{Q_{2i,j+1}[w]} },~~~~~~~~~~~\text{if}~f\Big{(}Q_{2i-1,j+1}[w]\Big{)}=f\Big{(}Q_{2i,j+1}[w]\Big{)}=1
          \end{aligned} \right.
\end{equation}

\begin{equation} \label{eq_decoding_posterior metric_g_1}
  Q_{i,j}[w] = \left\{
    \begin{aligned}
              &e,~~~~~~~~~~~~~~~~~~~~~~~~~~~~~~~~~~~~~~~~~~~~~~~~~~~~~~~~~~~~~~~~~~~~~~~~~~~~\text{if}~Q_{2i-1,j+1}[w]=Q_{2i,j+1}[w] \\
              & \begin{aligned}
                    &\overline{Q_{2i-1,j+1}[w]\overline{Q_{2i, j+1}[w]}+\overline{f\Big{(}Q_{2i-1,j+1}[w]\Big{)}}~\overline{Q_{2i, j+1}[w]}} \\
                    &+ \overline{f\Big{(}Q_{2i-1,j+1}[w]\Big{)}}Q_{2i, j+1}[w] + \overline{Q_{2i-1,j+1}[w]}~\overline{f\Big{(}Q_{2i, j+1}[w]\Big{)}}
                    \end{aligned},
                  ~~~\text{if}~f\Big{(}Q_{2i-1,j+1}[w]\Big{)}=\overline{f\Big{(}Q_{2i, j+1}[w]\Big{)}}  \\
              &Q_{2i,j+1}[w],~~~~~~~~~~~~~~~~~~~~~~~~~~~~~~~~~~~~~~~~~~~~~~~~~~~~~~~~~~~~~~~~~\text{if}~Q_{2i-1,j+1}[w]=\overline{Q_{2i,j+1}[w]}
\end{aligned} \right.
\end{equation}

\hrulefill
\end{figure*}

For $1\leqslant w\leqslant r$, if the estimated value $\hat{U}_{i-N/2,j}[w] = 0$, then $Q_{i,j}[w]$ is calculated by using Eq. (\ref{eq_decoding_posterior metric_g_0}) and if $\hat{U}_{i-N/2,j}[w] = 1$, the $Q_{i,j}[w]$ is computed by using Eq. (\ref{eq_decoding_posterior metric_g_1}).

The estimated variables are calculated successively in accordance with the following rules.
\begin{equation} \label{eq_decoding_estimed_values}
  \hat{U}_{i,j} = \left\{
    \begin{aligned}
      &\mathbf{0}~~~~~~~~~~~~~~~~~~~~~~~~~~~~~\text{if}~j=0~\text{and}~i\notin{\mathcal{I}}; \\
      &Q_{i,j}~~~~~~~~~~~~~~~~~~~~~~~~~~\text{if}~j=0~\text{and}~i\in{\mathcal{I}};\\
      &\hat{U}_{i/2,j-1}\star \hat{U}_{i/2+N/2,j-1}~~~\text{if}~j\neq 0~\text{and}~i~\text{even};\\
      &\hat{U}_{(i+1)/2+N/2,j-1}~~~~~~~~~~~\text{if}~j\neq 0~\text{and}~i~\text{odd}.
\end{aligned} \right.
\end{equation}

We define a metric vector $K_{i,j}$ of the posterior packet $Q_{i,j}$ as
\begin{equation}\label{Eq_define_posterior_metric}
  K_{i,j} \triangleq f(Q_{i,j}) = \Big{(}f(Q_{i,j}[1]),~...~,~f(Q_{i,j}[r])\Big{)}
\end{equation}
which can be used to measure the past trajectory of the decoded path.

The pSCL decoding algorithm is required to preserve $L$ survival paths with the max path metric at each decision stage of each packet \cite{Polarcode_magazine:2014}\cite{SClist_decoding_Tal_Vardy:2015}. For $1 \leqslant i \leqslant N$ and $1\leqslant {\ell}\leqslant L$, the ${\ell}$th path metric vector of  estimated packet vector $\hat{U}_1^i$ is denoted as $M_{\ell}^{(i)} \triangleq \big{(}m_{\ell}[1], m_{\ell}[2],...,m_{\ell}[r]\big{)}$. For $1\leqslant w\leqslant r$, each element $m_{\ell}[w]^{(i)}$ can be recursively computed as \cite{LLR_SCL_Burg:2015}

\begin{equation*}
m_{\ell}[w]^{(i)} = m_{\ell}[w]^{(i-1)} + K_{i,0}[w].
\end{equation*}

At the end of pSCL decoding, the survival path with the maximum value from the $L$ path metric is selected as the decoding result.

\subsection{Throughput Analysis of the PSA}

In this section, an upper and a lower non-asymptotic throughput bounds and the asymptotic throughput for the PSA schemes using the pSC decoding are investigated.

\subsubsection{For the case of finite $N$}
In the finite $N$ slots case, the non-asymptotic throughput $T$ for PSA schemes using the pSC decoding is evaluated by an upper bound and a lower bound. First, we define the error events as:
\begin{equation}\label{Eq_block_error_event}
  \mathcal{E} \triangleq \Big{\{} (U_1^N,Y_1^N) \subset ({\mathcal{X}}^{r} \times {\mathcal{Y}}^r)^N: \hat{U}_\mathcal{I}\neq U_{\mathcal{I}}\Big{\}}
\end{equation}
where $\hat{U}_\mathcal{I}$ denoted the output of pSC decoding.

\begin{theorem}\label{theorem_throughput_upper_lower_bound_of_pSC}
An upper bound and a lower bound of the throughput $T$ for PSA schemes using the pSC decoding are
\begin{equation}\label{upper_and_lower_bound_of_pSC}
  \frac{M}{N} \bigg{(} 1-\sum_{i \in \mathcal{I}} \frac{Z\big{(}S_N^{(i)}\big{)}}{r} \bigg{)} \leqslant T \leqslant \frac{M}{N} \bigg{(} 1-\max_{i \in \mathcal{I}} \frac{Z\big{(}S_N^{(i)}\big{)}}{r} \bigg{)}.
\end{equation}
\end{theorem}
\begin{proof}
From \emph{Lemma} \ref{lemma_SEC_product_compound_channel}, the SEC $S$ guarantees each component channel $W$ suffers from the identical noise realization. And as mentioned in \cite{polarcode_Arikan}, the probability of error event $P(\mathcal{E})$ for the pSC decoding with an upper bound and a lower bound, we rewrite them as
\begin{equation}\label{Eq_p_SC_upper_bound}
  P(\mathcal{E}) \leqslant \sum_{i \in \mathcal{I}} \frac{Z\Big{(}S_N^{(i)}\Big{)}}{r}
\end{equation}
and
\begin{equation}\label{Eq_p_SC_lower_bound}
  P(\mathcal{E}) \geqslant \max_{i \in \mathcal{I}} \frac{Z\Big{(}S_N^{(i)}\Big{)}}{r}
\end{equation}
where the parameters $Z(S_N^{(i)})/r$ are reliability metric (with respect to packets) of synthetic channel $S_N^{(i)}$, $i\in \mathcal{I}$.

Recall that $P_u = 1-P(\mathcal{E})$, and replace Eq. (\ref{Eq_p_SC_upper_bound}) and (\ref{Eq_p_SC_lower_bound}) into Eq. (\ref{eq_throughput_efficiency}), the Eq. (\ref{upper_and_lower_bound_of_pSC}) of throughput $T$ holds. Therefore, the proof of this \emph{Theorem} completes.
\end{proof}

\begin{figure}[htbp]
\centerline{\includegraphics[scale=0.5]{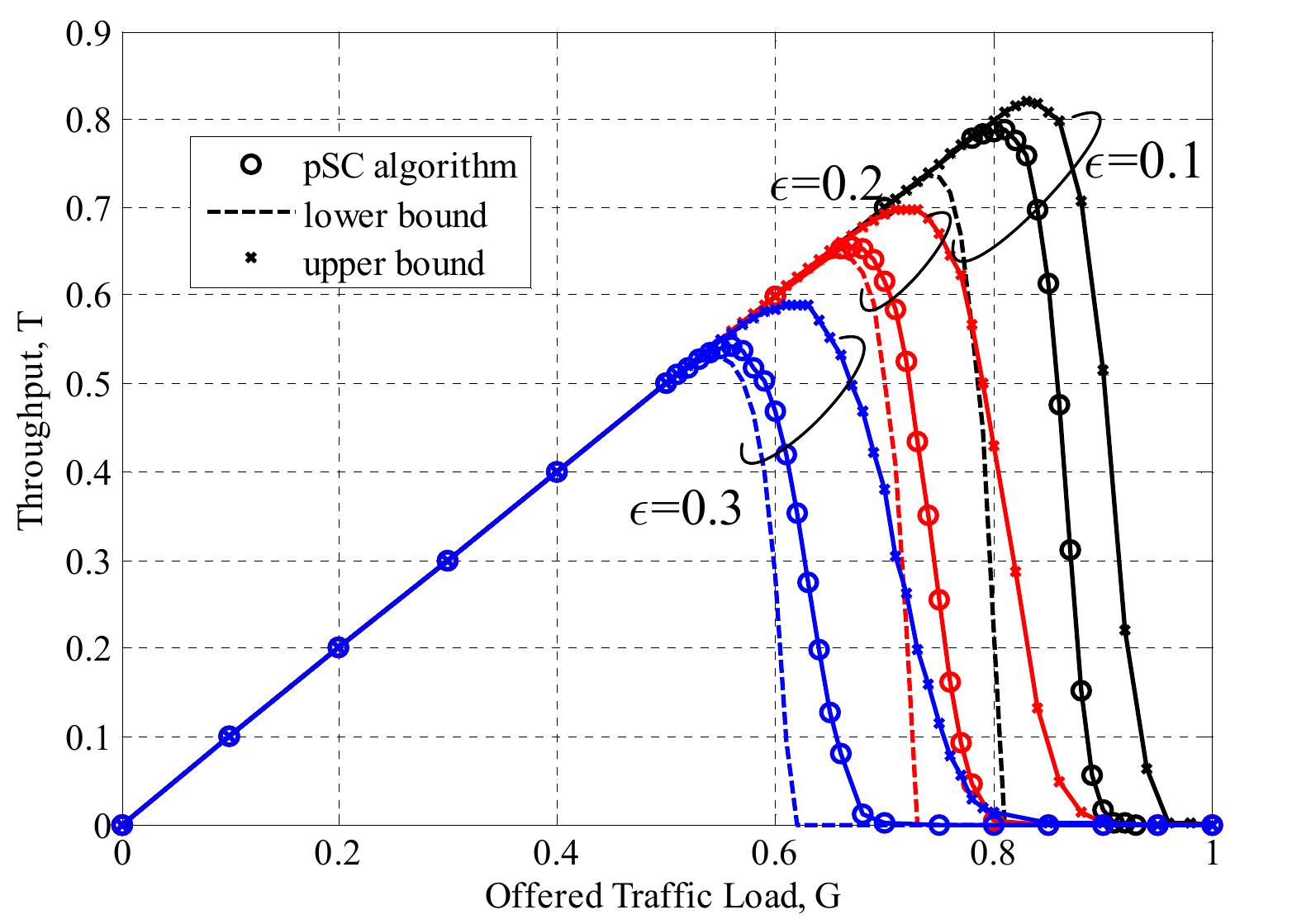}}
\caption{The upper bound and lower bound of the throughput versus traffic load for PSA schemes over the SECs using the pSC decoding under different slot erasure rates $\epsilon=0.1,0.2,0.3$ and $N=1024$.}
\label{SC_with_up_low_bound}
\end{figure}

The curves of the throughput for PSA schemes over SECs using the pSC algorithm and the upper/lower bounds under different SEP values are shown in Fig. (\ref{SC_with_up_low_bound}). When $\epsilon$ increases from $0.1$ to $0.3$, it can be seen that the upper bound is getting looser, but the lower bound is getting tighter. However, in general, the upper/lower bound is still relatively loose.

\subsubsection{For the case of infinite $N\rightarrow  \infty$}
In the infinite $N$ slots case, the asymptotic throughput $T$ for PSA schemes using the pSC decoding can be evaluated by \emph{Theorem} \ref{theorem_throughput_for_infinite slots case}.
\begin{theorem}\label{theorem_throughput_for_infinite slots case}
With the SPA-v algorithm, the asymptotic throughput of the PSA scheme over SECs is
\begin{equation}
  T_{a} = 1-\epsilon
\end{equation}

\end{theorem}
\begin{proof}
In the proposed PSA schemes with the SPA-v algorithm using the variable SEP, the active users select synthetic channels with higher capacity as their SPs. From \emph{Theorem} \ref{theorem_SEC_polarization_effect}, as the number $N$ of slots goes to infinity through powers of two, the synthetic channels which as SPs are almost perfect. Corresponding, the asymptotic recovery probability of user packets in the BS will achieve the capacity of the slot erasure channel, that is $P_u = 1-\epsilon$.

Accordingly, when the offered traffic load $G=M/N<(1-\epsilon)$, as the number $N$ goes to infinity, the SC decoding will correctly recovery the packets transmitted from the active users in the PSA scheme. Consequently, we make a conclusion that the asymptotic throughput of the proposed PSA scheme is $T_{a} = 1-\epsilon$.
\end{proof}

The simulation results for the PSA scheme over SECs will be shown in the next section.

\section{Simulation Results}\label{section_simulation_results}

In this section, we will evaluate the throughput of the PSA schemes over SECs with different parameters.

\begin{figure}[htbp]
\centerline{\includegraphics[scale=0.5]{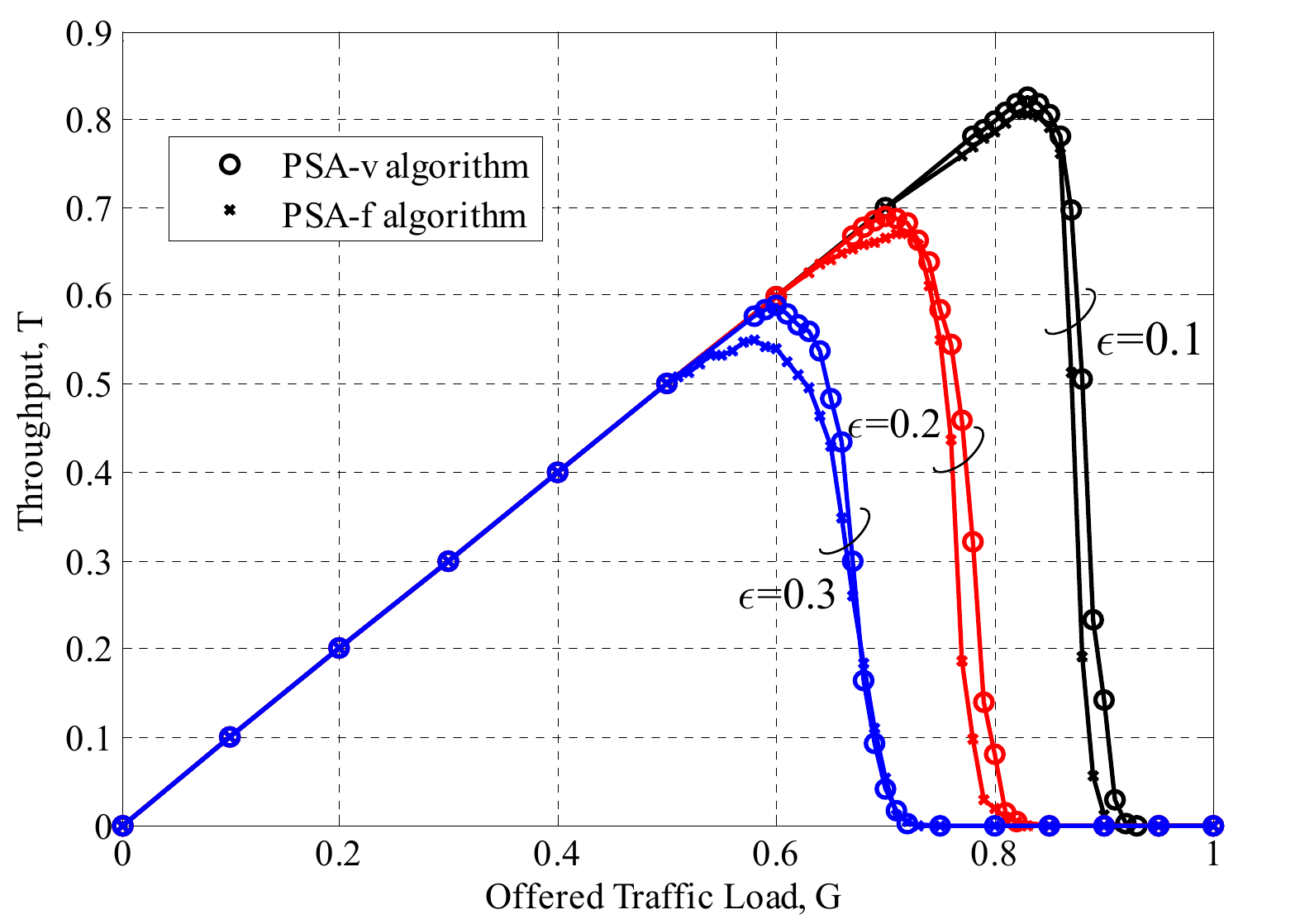}}
\caption{For the different SPA methods, the throughput $T$ versus offered traffic load $G$ for the PSA schemes transmit over SECs with SEP values $\epsilon=0.1,0.2,0.3$, and $N=1024$ using the pSCL ($L=16$) decoding.}
\label{Throughout_PSA_N_1024_e_0_1_0_2_0_3_SPA_v_SPA_f_algorithm}
\end{figure}

Fig. \ref{Throughout_PSA_N_1024_e_0_1_0_2_0_3_SPA_v_SPA_f_algorithm} shows the throughput curves of the PSA schemes over SECs with the two SPA methods under the same pSCL ($L=16$) decoding and $N=1024$. Obviously, for different SEP values, it reads that the maximum throughput $T^{*}$ of PSA schemes with the SPA-v algorithm is always higher than that with the SPA-f algorithm. This observation validates the previous analysis in section \ref{section_Proposed_PSA_schemes_over_SECs}. Therefore, the SPA-v algorithm is used in the following performance evaluation.

\begin{figure}[htbp]
\centerline{\includegraphics[scale=0.5]{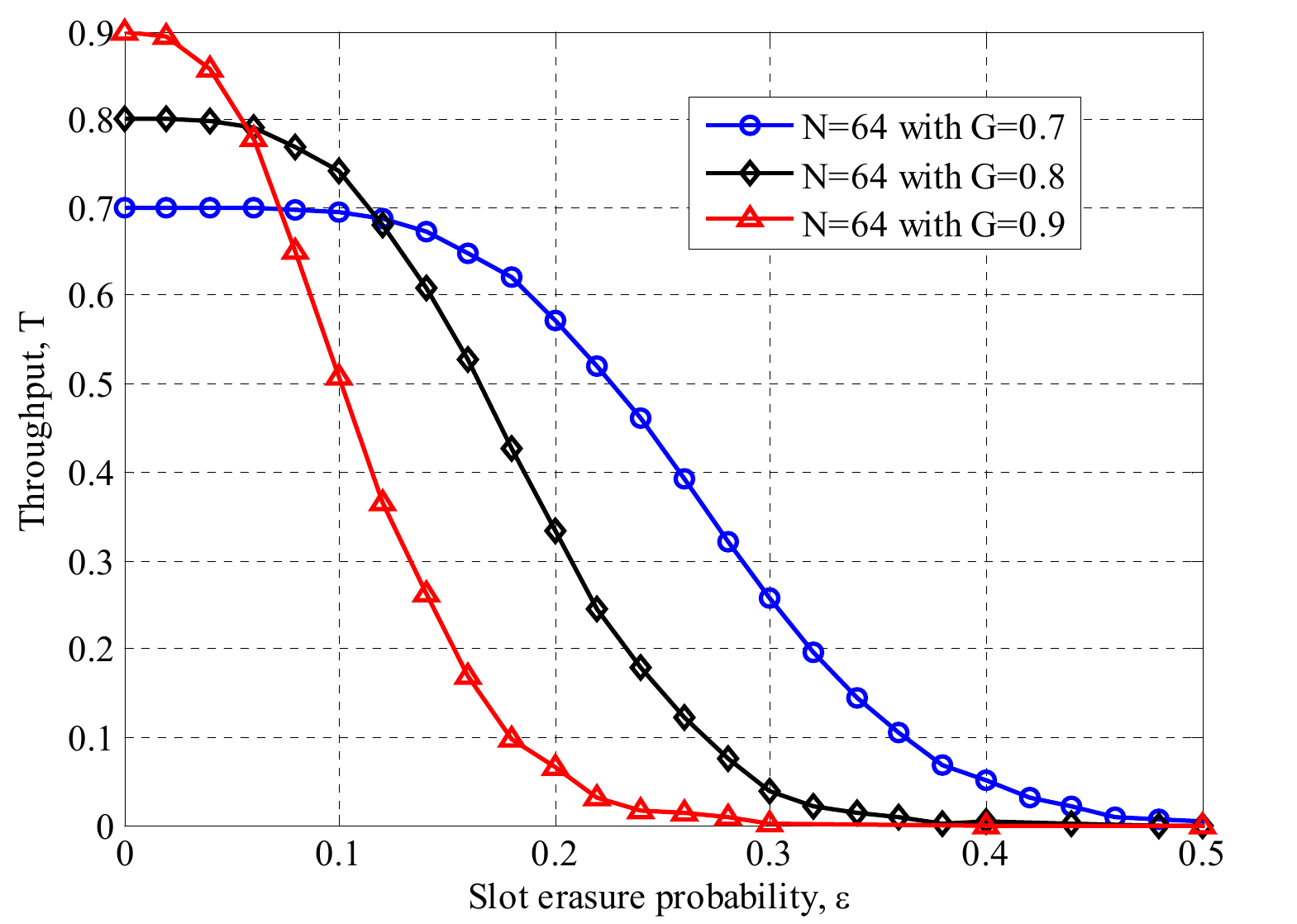}}
\caption{Given the pSC decoding and $N=64$, the throughput versus different SEPs $\epsilon$ with the traffic load $G=0.7$, $0.8$ and $0.9$.}
\label{Throughout_PSA_N_64_vs_e_with_G_0_7_0_8_0_9}
\end{figure}

In the finite $N$ case, with the fixed traffic load, the throughput of the PSA scheme over SECs is affected by the SEP values. It can be seen from Fig.\ref{Throughout_PSA_N_64_vs_e_with_G_0_7_0_8_0_9}, with the SEP value changes from $0$ to $0.1$, the throughput is stable for the fixed traffic load $G=0.7$. However, the throughput is almost halved with the fixed traffic load $G=0.9$. In other words, the higher with respect to the traffic load, the more throughput sensible with the variation of slot erasure probability.

\begin{figure}[htbp]
\centerline{\includegraphics[scale=0.5]{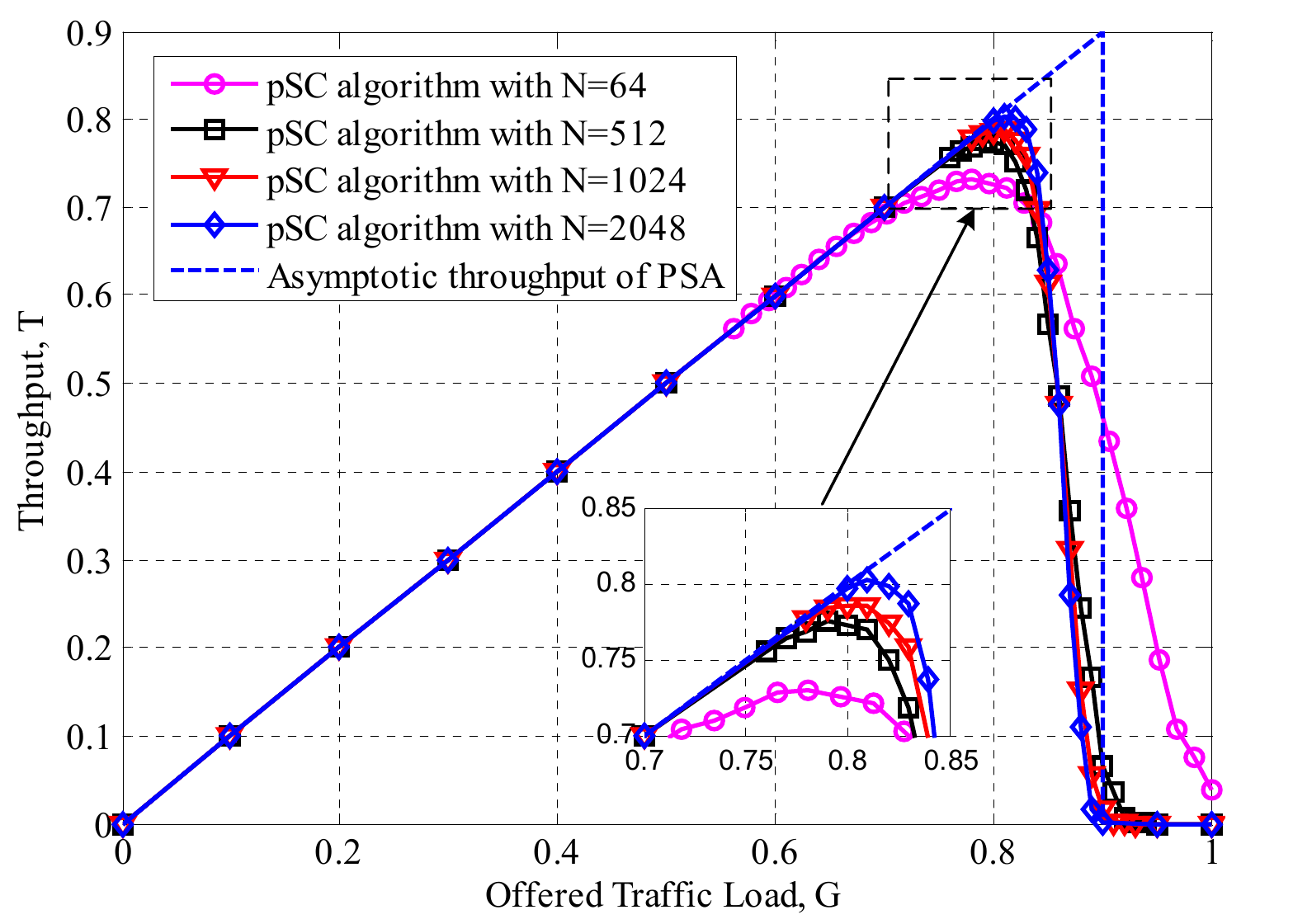}}
\caption{Under the SPA-v algorithm with a SEP $\epsilon=0.1$, the throughput versus traffic load for the different $N=64,512,1024,2048$ using the pSC decoding algorithm.}
\label{Throughtout_PSA_N_64_512_1024_2048_e_0_1}
\end{figure}

Given the identical pSC decoding algorithm and the SEP value, it can be seen from Fig. \ref{Throughtout_PSA_N_64_512_1024_2048_e_0_1} that the maximum throughput is $T^{*}=0.73$ for $N=64$, $T^{*}=0.77$ for $N=512$, $T^{*}=0.79$ for $N=1024$ and $T^{*}=0.80$ for $N=2048$. That is, the throughput of the PSA scheme can be improved with the increasing of $N$. This observation can be explained by the polarization effect becomes more significant with the more slots within a slot-frame of the PSA scheme.

\begin{figure}[htbp]
\centerline{\includegraphics[scale=0.5]{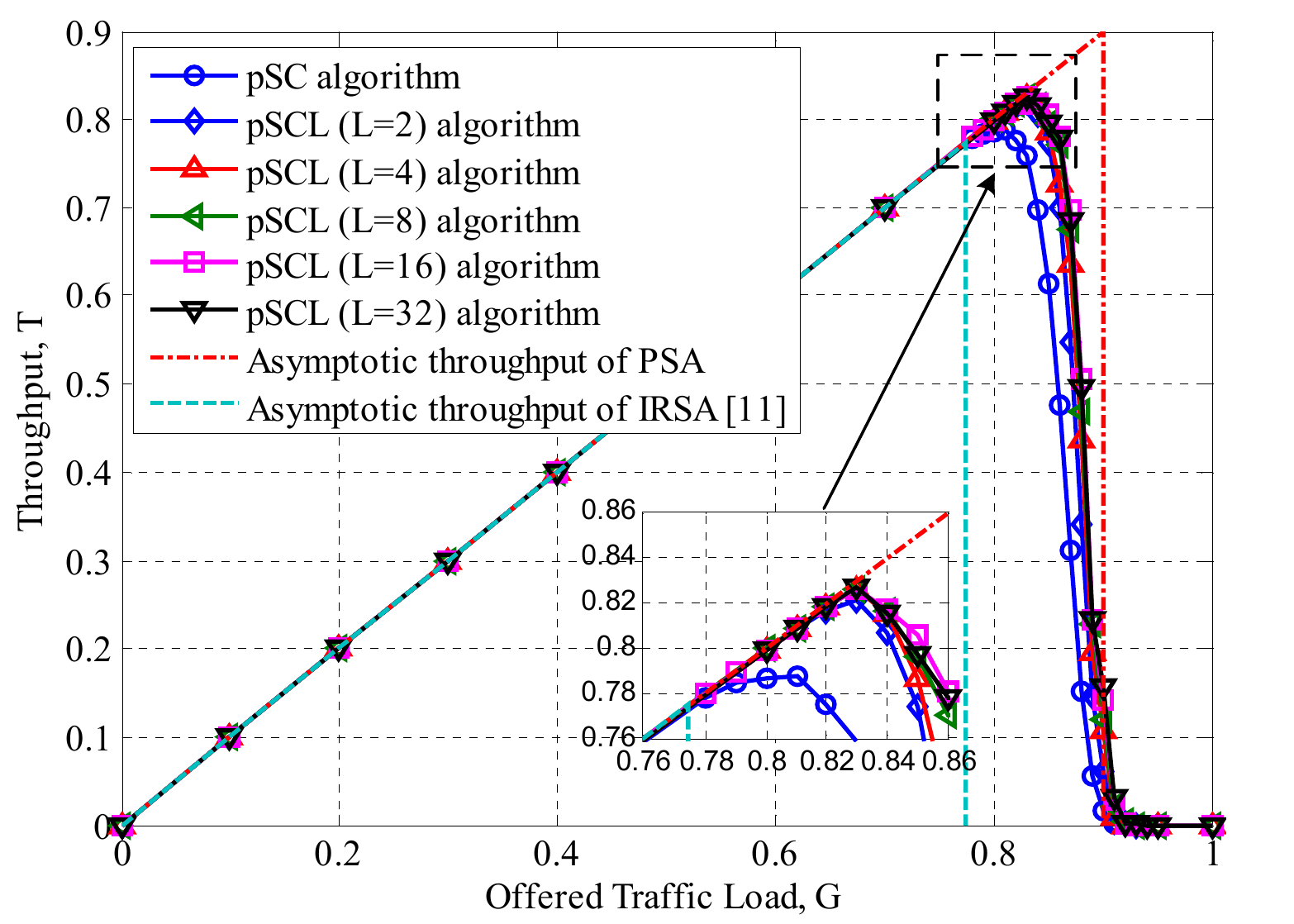}}
\caption{Under the SPA-v algorithm with SEP $\epsilon=0.1$ and $N=1024$, the throughput versus traffic load for the PSA schemes using the pSC/SCL decoding algorithm.}
\label{Throughtout_PSA_N_1024_SCL_L_1_2_4_8_16_32}
\end{figure}

It can be seen from Fig. \ref{Throughtout_PSA_N_1024_SCL_L_1_2_4_8_16_32} that the proposed PSA scheme can achieve an improved throughput with the pSC/SCL decoding algorithm over the traditional IRSA scheme. Given the SEP $\epsilon=0.1$ and $N=1024$, the maximum throughput $T^{*}=0.79$ of the PSA scheme using the pSC decoding exceeds the the asymptotic threshold of the traditional IRSA scheme \cite{CSAErasureChannels:Sun}. Furthermore, the maximum throughput of the PSA scheme using the pSCL decoding can be further improved by increasing $L$. Compared to the traditional IRSA scheme \cite{CSAErasureChannels:Sun}, the asymptotic throughput in the proposed PSA scheme is increased about $0.11$ packets/slot at SEP $\epsilon=0.1$.

How to eliminate the gap between the asymptotic throughput and the actual throughput under the finite $N$ slots case is an interesting issue, other methods should be sought no more than only rely on increasing the list $L$ of pSCL decoding. Just like the CRC-aided SCL decoding of bit-oriented polar codes \cite{CRC_aided_Kai:2012}, the gap may be narrowed by utilizing the prior information about the integrity check of each packet.

\section{Conclusions}\label{section_conclusions}
In this paper, we proposed the PSA schemes over slot erasure channels by using the polar coding to construct identical SP sets in each active user and the BS. Relative to the traditional repetition slotted ALOHA scheme, handling pointers process is avoided in the PSA schemes because of using the identical SP sets. We provided a theoretical analysis framework of the PSA schemes. Based on the packet-oriented operation for the overlap packets when they conflict in a slot, we proved that this operation guarantees the packet-based polarization transform maintains the polarization phenomenon regardless of the length of bits within the packet. Guided by the packet-based polarization, the SPA-v and the SPA-f algorithm for the SP assignment were developed. Finally, the pSC and the pSCL decoding algorithms were introduced. For the case of finite $N$, an upper bound and a lower bound for the PSA schemes using the pSC decoding were investigated. And more, for the case of infinite $N$, the asymptotic throughput of the PSA schemes was also analyzed. Furthermore, the simulation results were given to verify that the proposed PSA scheme can achieve an improved throughput with the pSC/SCL decoding algorithm over the traditional IRSA scheme. How to approach to the asymptotic throughput of PSA in the case of finite $N$ is an interesting issue. The prior information of integrity check of each user's packet can be utilized which will be investigated for the future work of the coded PSA schemes.

\section*{Appendix}

\subsection{Proof of Lemma 2} \label{subsection_proof_of_lemma_2}
\begin{proof}
 By using Eq. (\ref{eq_def_packet_oriented_operation}), we obtain that
\begin{align*}
       & S_{2}(Y_1,Y_2|U_1,U_2) \\
          = & S(Y_1|U_1 \boxplus U_2)S(Y_2|U_2) \\
          = & S\big{(}(y_1)_1^r|(u_1)_1^r \boxplus (u_2)_1^r\big{)} S\big{(}(y_2)_1^r|(u_2)_1^r\big{)} \\
          = & S\big{(}(y_1)_1^r|u_{1,1}\oplus u_{2,1},...,u_{1,r}\oplus u_{2,r}\big{)} S\big{(}(y_2)_1^r|(u_2)_1^r\big{)} \\
          = & \big{(}W(y_{1,1}|u_{1,1}\oplus u_{2,1})\times...\times W(y_{1,r}|u_{1,r}\oplus u_{2,r})\big{)} \\
            & \big{(} W(y_{2,1}|u_{2,1})\times...\times W(y_{2,r}|u_{2,r}) \big{)} \\
          = & \big{(}W(y_{1,1}|u_{1,1}\oplus u_{2,1})W(y_{2,1}|u_{2,1})\big{)} \times... \\
            & \times \big{(} W(y_{1,r}|u_{1,r}\oplus u_{2,r})W(y_{2,r}|u_{2,r}) \big{)} \\
          = & W_2(y_{1,1},y_{2,1}|u_{1,1},u_{2,1}) \times...\times W_2(y_{1,r},y_{2,r}|u_{1,r},u_{2,r})\\
          = & \prod_{i=1}^{r}{W_2(y_{1,i},y_{2,i}|u_{1,i},u_{2,i})}.
\end{align*}
And then, Eq. (\ref{Eq_S2_Equivalent}) holds. Furthermore, substituting Eq. (\ref{Eq_S2_Equivalent}) into Eq. (\ref{eq_def_S+}), we obtain that
  \begin{align*}
  & S^{+}(Y_1,Y_2,U_1|U_2) \\
  = &  {\frac{1}{2^r} S_2(Y_1,Y_2|U_1,U_2)} \\
  = &  {\frac{1}{2^r} \prod_{i=1}^{r}{W_2(y_{1,i},y_{2,i}|u_{1,i},u_{2,i})} }\\
  = &  { \prod_{i=1}^{r}\frac{1}{2}{W_2(y_{1,i},y_{2,i}|u_{1,i},u_{2,i})} }\\
  = &  \prod_{i=1}^{r} W^{+}(y_{1,i},y_{2,i},u_{1,i}|u_{2,i}).
 \end{align*}
Obviously, Eq.(\ref{Eq_S+_Equivalent}) holds.

With the same approach, substituting Eq. (\ref{Eq_S2_Equivalent}) into Eq. (\ref{eq_def_S-}), it can be obtained that
\begin{align*}
  &  S^{-}(Y_1,Y_2|U_1) \\
  = & \sum_{U_2 \in \mathcal{X}^r } {\frac{1}{2^r} S_2(Y_1,Y_2|U_1,U_2)} \\
  = & \frac{1}{2^r} \sum_{U_2 \in \mathcal{X}^r } {\prod_{i=1}^{r} W_{2}(y_{1,i},y_{2,i}|u_{1,i},u_{2,i})} \\
  = & \frac{1}{2^r} \sum_{\{u_{2,1} \in \mathcal{X}\}\cup ...\cup \{u_{2,r}\in \mathcal{X}\}} {\prod_{i=1}^{r} W_{2}(y_{1,i},y_{2,i}|u_{1,i},u_{2,i})} \\
  = & \frac{1}{2^r} \bigg{(} \sum_{u_{2,1} \in \mathcal{X} } \prod_{i=1}^{r} W_{2}(y_{1,i},y_{2,i}|u_{1,i},u_{2,i})\times ... \\
    & \times \sum_{u_{2,r} \in \mathcal{X} } {\prod_{i=1}^{r} W_{2}(y_{1,i},y_{2,i}|u_{1,i},u_{2,i})} \bigg{)} \\
  = & \frac{1}{2^r} \bigg{(} \sum_{u_{2,1} \in \mathcal{X} } W_{2}(y_{1,1},y_{2,1}|u_{1,1},u_{2,1})\times ... \\
    & \times \sum_{u_{2,r} \in \mathcal{X} } { W_{2}(y_{1,r},y_{2,r}|u_{1,r},u_{2,r})} \bigg{)} \\
  = & \bigg{(} \frac{1}{2} \sum_{u_{2,1} \in \mathcal{X} } W_{2}(y_{1,1},y_{2,1}|u_{1,1},u_{2,1})\bigg{)} \times ... \\
    & \times \bigg{(}\frac{1}{2} \sum_{u_{2r} \in \mathcal{X} } { W_{2}(y_{1,r},y_{2,r}|u_{1,r},u_{2,r})} \bigg{)} \\
  = & W^{-}(y_{1,1},y_{2,1}|u_{1,1}) \times...\times W^{-}(y_{1,r},y_{2,r}|u_{1,r})\\
  = & \prod_{i=1}^{r}{W^{-}(y_{1,i},y_{2,i}|u_{1,i})}.
\end{align*}
So Eq. (\ref{Eq_S-_Equivalent}) holds. Therefore, the proof of \emph{Lemma} \ref{lemma_SEC_product_compound_channel} completes.

\end{proof}

\subsection{Proof of Lemma 3}
\begin{proof}
When $N=1$, the synthetic channel $S_N^{(i)}$ degrades to $W$. Obviously, from \emph{Lemma} $1$ and \emph{Remark} $1$, Eq. (\ref{eq_capcaity_combined_polarized_channel}) holds. When $N>1$, for any $1\leqslant i\leqslant N$, from \emph{Lemma} $2$, \emph{Remark} $2$ and \cite{parallel channels:1976}, using Eq. (\ref{Eq_relation_S_and_W}), then

\begin{equation*}
 \begin{aligned}
I\Big{(}S_N^{(j)}\Big{)} = &I\Big{(}\prod_{i=1}^{r} {W^{(j)}_{N}(y_{1,i},...,y_{N,i},u_{1,i},...u_{j-1,i}|u_{j,i})}\Big{)} \\
                         = &I\Big{(}{W^{(j)}_{N}(y_{1,1},...,y_{N,1},u_{1,1},...u_{j-1,1}|u_{j,1})} \times ...\\
                           &\times{W^{(j)}_{N}(y_{1,r},...,y_{N,r},u_{1,r},...u_{j-1,r}|u_{j,r})} \Big{)} \\
                         = &I\Big{(}{W^{(j)}_{N}(y_{1,1},...,y_{N,1},u_{1,1},...u_{j-1,1}|u_{j,1})}\Big{)} +...\\
                           &+ I\Big{(}{W^{(j)}_{N}(y_{1,r},...,y_{N,r},u_{1,r},...u_{j-1,r}|u_{j,r})} \Big{)} \\
                         = &rI\Big{(}W_N^{(j)}\Big{)}.
 \end{aligned}
\end{equation*}
the proof of \emph{Lemma} \ref{lemma_capcaity_combined_polarized_channel} completes.
\end{proof}

\subsection{Proof of Theorem $1$}
\begin{proof}
For any $1\leqslant j \leqslant N/2$, $N=2^n$, the capacity and the Bhattacharyya parameter of synthetic channels for the BEC $W$ with a erasure probability $\epsilon$ are computed using the recursive relations \cite{polarcode_Arikan} as
\begin{equation*}
    \left\{
   \begin{aligned}
   I\Big{(}W_N^{(2j-1)}\Big{)} & = I\Big{(}W_{N/2}^{(j)}\Big{)}^2 \\
     I\Big{(}W_N^{(2j)}\Big{)} & = 2I\Big{(}W_{N/2}^{(j)}\Big{)} - I\Big{(}W_{N/2}^{(j)}\Big{)}^2
   \end{aligned} \right.\label{eq_capcaity_combined_bit_polarized_channel}
\end{equation*}

 \begin{equation*}
    \left\{
   \begin{aligned}
   Z\Big{(}W_N^{(2j-1)}\Big{)} & =  2Z\Big{(}W_{N/2}^{(j)}\Big{)} - Z\Big{(}W_{N/2}^{(j)}\Big{)}^2\\
     Z\Big{(}W_N^{(2j)}\Big{)} & =  Z\Big{(}W_{N/2}^{(j)}\Big{)}^2
   \end{aligned} \right.\label{eq_capcaity_combined_bit_polarized_channel}
\end{equation*}
with the initial values $I_0 = I\Big{(}W_1^{(1)}\Big{)} = 1-\epsilon$ and $Z_0 = Z\Big{(}W_1^{(1)}\Big{)} = \epsilon$.

From \emph{Lemma} \ref{lemma_capcaity_combined_polarized_channel}, the capacity of synthetic channel of SECs can be computed as
 \begin{equation*}
   \begin{aligned}
 I\Big{(}S_N^{(2j-1)}\Big{)} = &rI\Big{(}W_N^{(2j-1)}\Big{)} \\
                 = &rI\Big{(}W_{N/2}^{(j)}\Big{)}^2 \\
                 = &\frac{1}{r}\Big{(}rI\Big{(}W_{N/2}^{(j)}\Big{)}\Big{)}^2 \\
                 = &\frac{1}{r}I\Big{(}S_{N/2}^{(j)}\Big{)}^2
   \end{aligned}
\end{equation*}
and
\begin{equation*}
   \begin{aligned}
 I\Big{(}S_N^{(2j)}\Big{)} = &rI\Big{(}W_N^{(2j)}\Big{)} \\
                 = &r\bigg{(} 2I\Big{(}W_{N/2}^{(j)}\Big{)} - I\Big{(}W_{N/2}^{(j)}\Big{)}^2 \bigg{)}\\
                 = &2rI\Big{(}W_{N/2}^{(j)}\Big{)} - rI\Big{(}W_{N/2}^{(j)}\Big{)}^2 \\
                 = &2\bigg{(}rI\Big{(}W_{N/2}^{(j)}\Big{)}\bigg{)} -\frac{1}{r}\bigg{(} rI\Big{(}W_{N/2}^{(j)}\Big{)} \bigg{)}^2 \\
                 = &2I\Big{(}S_{N/2}^{(j)}\Big{)} -\frac{1}{r} I\Big{(}S_N^{(j)}\Big{)}^2.
   \end{aligned}
\end{equation*}
And then, Eq. (\ref{SEC_info_capacity}) holds. Using the same approach, the Bhattacharyya parameters for SEC synthetic channels are computed as
\begin{equation*}
   \begin{aligned}
 Z\Big{(}S_N^{(2j-1)}\Big{)} = &rZ\Big{(}W_N^{(2j-1)}\Big{)} \\
                             = &r\bigg{(} 2Z\Big{(}W_{N/2}^{(j)}\Big{)} - Z\Big{(}W_{N/2}^{(j)}\Big{)}^2 \bigg{)}\\
                             = &2rZ\Big{(}W_{N/2}^{(j)}\Big{)} - rZ\Big{(}W_{N/2}^{(j)}\Big{)}^2 \\
                             = &2\bigg{(}rZ\Big{(}W_{N/2}^{(j)}\Big{)}\bigg{)} -\frac{1}{r}\bigg{(} rZ\Big{(}W_{N/2}^{(j)}\Big{)} \bigg{)}^2 \\
                             = &2Z\Big{(}S_{N/2}^{(j)}\Big{)} -\frac{1}{r} Z\Big{(}S_N^{(j)}\Big{)}^2.
   \end{aligned}
\end{equation*}
and
\begin{equation*}
   \begin{aligned}
 Z\Big{(}S_N^{(2j)}\Big{)} = &rZ\Big{(}W_N^{(2j)}\Big{)} \\
                           = &rZ\Big{(}W_{N/2}^{(j)}\Big{)}^2 \\
                           = &\frac{1}{r}\Big{(}rZ\Big{(}W_{N/2}^{(j)}\Big{)}\Big{)}^2 \\
                           = &\frac{1}{r}Z\Big{(}S_{N/2}^{(j)}\Big{)}^2.
   \end{aligned}
\end{equation*}
 So Eq. (\ref{SEC_Z_parameter}) holds. The proof of \emph{Theorem} \ref{theorem_SEC_capacity_and_Z_parameter} completes.
\end{proof}





\ifCLASSOPTIONcaptionsoff
  \newpage
\fi


\begin{thebibliography}{99}

\bibitem{CRDSA}
E.~Casini, R.D.~Gaudenzi, and O.D.R.~Herrero, ``Contention resolution diversity slotted ALOHA (CRDSA): An enhanced random access scheme for satellite access packet networks," \emph{IEEE Trans. Wireless Commun.}, vol. 6, no. 4, pp. 1408--1419, Apr. 2007.

\bibitem{IRSA}
G.~Liva, ``Graph-based analysis and optimization of contention resolution diversity slotted ALOHA," \emph{IEEE Trans. Commun.}, vol. 59, no. 2, pp. 477--487, Feb. 2011.

\bibitem{CSA:Paolini}
E.~Paolini, G.~Liva and M.~Chiani, ``Coded Slotted ALOHA: A Graph-Based Method for Uncoordinated Multiple Access," \emph{IEEE Trans. Inf. Theory}, vol. 61, no. 12, pp. 6815--6832, Dec. 2015.

\bibitem{Asymptotic_Perfo_of_CSA_CL:2018}
C. Stefanovic, E. Paolini and G. Liva, ``Asymptotic Performance of Coded Slotted ALOHA With Multipacket Reception," \emph{IEEE Commun. Lett.}, vol. 22, no. 1, pp. 105--108, Jan. 2018.

\bibitem{Non_asymptotic_CSA_ISIT:2019}
M. Fereydounian, X. Chen, H. Hassani and S. S. Bidokhti, ``Non-asymptotic Coded Slotted ALOHA," in \emph{IEEE Int. Sym. Info. Theory (ISIT)}, Paris, pp. 111--115, July, 2019.

\bibitem{err_floor_CSA}
M.~Ivanov, F.~Brannstrom, A. Graell i Amat and P. Popovski, ``Error Floor Analysis of Coded Slotted ALOHA Over Packet Erasure Channels," \emph{IEEE Commun. Lett.}, vol. 19, no. 3, pp. 419--422, Mar. 2015.

\bibitem{CSAErasureChannels:Sun}
Z.~Sun, Y.~Xie, J.~Yuan and T.~Yang, ``Coded Slotted ALOHA for Erasure Channels: Design and Throughput Analysis,"
\emph{IEEE Trans. Commun.}, vol. 65, no. 11, pp. 4817--4830, Nov. 2017.

\bibitem{CSA_magz:Paolini}
E.~Paolini, C.~Stefanovic, G.~Liva and P.~Popovski, ``Coded random access: applying codes on graphs to design random access protocols," \emph{IEEE Commun. Magazine}, vol. 53, no. 6, pp. 144--150, Jun. 2015.

\bibitem{Grouptest_shortpacket_ISIT:2019}
H.A.Inan, S. Ahn, P. Kairouzy and A. Ozgur, ``A Group Testing Approach to Random Access for Short-Packet Communication," in \emph{IEEE Int. Sym. Info. Theory (ISIT)}, Paris, pp. 96--100, July, 2019.

\bibitem{Enhanc_CRDSA_Power_Div_TVT:2018}
S. Alvi, S. Durrani and X. Zhou, ``Enhancing CRDSA With Transmit Power Diversity for Machine-Type Communication," \emph{IEEE Trans. Veh. Technol.}, vol. 67, no. 8, pp. 7790--7794, Aug. 2018.

\bibitem{PLNetCoding_for_RA_TVT:2019}
Z. Sun, L. Yang, J. Yuan and D. W. K. Ng, "Physical-Layer Network Coding Based Decoding Scheme for Random Access," \emph{IEEE Trans. Veh. Technol.}, vol. 68, no. 4, pp. 3550--3564, April 2019.

\bibitem{polarcode_Arikan}
E. Ar{\i}kan, ``Channel polarization: A method for constructing capacityachieving codes for symmetric binary-input memoryless channels," \emph{IEEE Trans. Info. Theory}, vol. 55, no. 7, pp. 3051--3073, July 2009.

\bibitem{beacon}
Y.~Ji, C.~Bockelmann and A.~Dekorsy, ``Numerical analysis for joint PHY and MAC perspective of Compressive Sensing Multi-User Detection with coded random access," In \emph{IEEE Int. Conf. Commun. Workshops (ICC Workshops)}, Paris, pp. 1018--1023, 2017.

\bibitem{product_channel_Shannon:1956}
C. Shannon, ``The zero error capacity of a noisy channel," \emph{IRE Trans. Info. Theory}, vol. 2, no. 3, pp. 8--19, September 1956.

\bibitem{parallel channels:1976}
Y. Horibe, ``Product-sum use of parallel channels (Corresp.)," \emph{IEEE Trans. Info. Theory}, vol. 22, no. 4, pp. 475--476, July 1976.

\bibitem{polar_DE_mori:2009}
R. Mori and T. Tanaka, ``Performance of Polar Codes with the Construction using Density Evolution," \emph{IEEE Commun. Lett.}, vol. 13, no. 7, pp. 519--521, July 2009.

\bibitem{polarcode_BEC_decoding:2018}
A. Balatsoukas-Stimming and A. Burg, ``Faulty Successive Cancellation Decoding of Polar Codes for the Binary Erasure Channel," \emph{IEEE Trans. Commun.}, vol. 66, no. 6, pp. 2322--2332, June 2018.

\bibitem{polar_decoding_multi_layer:2013}
A. Pamuk and E. Ar{\i}kan, ``A two phase successive cancellation decoder architecture for polar codes," \emph{IEEE Int. Sym. Info. Theory (ISIT)}, Istanbul, 2013, pp. 957--961.

\bibitem{Polarcode_magazine:2014}
K. Niu, K. Chen, J. Lin and Q. T. Zhang, ``Polar codes: Primary concepts and practical decoding algorithms," \emph{IEEE Commun. Magazine}, vol. 52, no. 7, pp. 192--203, July 2014.

\bibitem{SClist_decoding_Tal_Vardy:2015}
I. Tal and A. Vardy, ``List Decoding of Polar Codes," \emph{IEEE Trans. Info. Theory}, vol. 61, no. 5, pp. 2213--2226, May 2015.

\bibitem{LLR_SCL_Burg:2015}
A. Balatsoukas-Stimming, M. B. Parizi and A. Burg, ``LLR-Based Successive Cancellation List Decoding of Polar Codes," \emph{IEEE Trans. Signal Processing}, vol. 63, no. 19, pp. 5165--5179, Oct.1, 2015.

\bibitem{CRC_aided_Kai:2012}
K. Niu and K. Chen, ``CRC-aided decdoing of polar codes," \emph{IEEE Commun. Letter}, vol. 16, no. 10, pp.1668--1671, 2012.



\end{thebibliography}
\end{document}